\documentclass[11pt,english,onecolumn,draftcls]{IEEEtran}
\usepackage[T1]{fontenc}
\usepackage[latin9]{inputenc}
\usepackage{float}
\usepackage{textcomp}
\usepackage{amsmath}
\usepackage{graphicx}
\usepackage{amssymb}
\usepackage{esint}

\makeatletter

\newcommand{\lyxmathsym}[1]{\ifmmode\begingroup\def\b@ld{bold}
  \text{\ifx\math@version\b@ld\bfseries\fi#1}\endgroup\else#1\fi}

\providecommand{\tabularnewline}{\\}

\newtheorem{example}{Example}
\newtheorem{definitn}{Definition}
\newtheorem{lemma}{Lemma}
\newtheorem{thm}{Theorem}
\newtheorem{remrk}{Remark}

\usepackage{wrapfig}
\usepackage{tikz}  
\usepackage{stmaryrd}
\usepackage{MnSymbol}
\usepackage{bbm}
\usepackage{pgfplots}
\usetikzlibrary{plotmarks}
\DeclareMathAlphabet{\mathcal}{OMS}{cmsy}{m}{n}

\newcommand{\boxastlimits}{\operatornamewithlimits{
\boxast}}
\newcommand{\oastlimits}{\operatornamewithlimits{\oast}}

\makeatother

\usepackage{babel}

\begin{document}

\title{Spatially-Coupled Codes and Threshold Saturation on Intersymbol-Interference
Channels}

\author{Phong S. Nguyen, Arvind Yedla, Henry D. Pfister, and Krishna R. Narayanan%
\thanks{This material is based upon work supported by the National Science
Foundation under Grant No. 0747470. The work of P. Nguyen was also
supported in part by a Vietnam Education Foundation fellowship. Any
opinions, findings, conclusions, or recommendations expressed in this
material are those of the authors and do not necessarily reflect the
views of the National Science Foundation.%
}\\
{\normalsize Department of Electrical and Computer Engineering,
Texas A\&M University }\\
{\normalsize College Station, TX 77840, U.S.A. }\\
{\normalsize \{psn, yarvind, hpfister, krn\}@tamu.edu \vspace{-5mm}}}
\maketitle
\begin{abstract}
Recently, it has been observed that terminated low-density-parity-check
(LDPC) convolutional codes (or spatially-coupled codes) appear to
approach capacity universally across the class of binary memoryless
channels. This is facilitated by the {}``threshold saturation''
effect whereby the belief-propagation (BP) threshold of the spatially-coupled
ensemble is boosted to the maximum a-posteriori (MAP) threshold of
the underlying constituent ensemble. 

In this paper, we consider the universality of spatially-coupled codes
over intersymbol-interference (ISI) channels under joint iterative
decoding. More specifically, we empirically show that threshold saturation
also occurs for the considered problem. This can be observed by first
identifying the EXIT curve for erasure noise and the GEXIT curve for
general noise that naturally obey the general area theorem. From these
curves, the corresponding MAP and the BP thresholds are then numerically
obtained. With the fact that regular LDPC codes can achieve the symmetric
information rate (SIR) under MAP decoding, spatially-coupled codes
with joint iterative decoding can universally approach the SIR of
ISI channels. For the dicode erasure channel, Kudekar and Kasai recently
reported very similar results based on EXIT-like curves. \end{abstract}

\section{Introduction}

Irregular low-density parity-check (LDPC) codes can be carefully designed
to achieve the capacity of the binary erasure channel (BEC) \cite{Luby-it01}
and closely approach the capacity of general binary-input symmetric-output
memoryless (BMS) channels \cite{Chung-comlett01} under belief-propagation
(BP) decoding. LDPC convolutional codes, which were introduced in
\cite{Felstrom-it99} and shown to have excellent BP thresholds in
\cite{Sridharan-aller04,Lentmaier-isit05}, have recently been observed
to \emph{universally} approach the capacity of various channels. The
fundamental mechanism behind this is explained well in \cite{Kudekar-it11},
where it is proven analytically for the BEC that the BP threshold
of a particular spatially-coupled ensemble converges to the maximum
a-posteriori (MAP) threshold of the underlying ensemble. A similar
result was also observed independently in \cite{Lentmaier-isit10}
and stated as a conjecture. Such a phenomenon is now called {}``threshold
saturation via spatial coupling'' and has also been empirically observed
for general BMS channels \cite{Kudekar-istc10}. In fact, threshold
saturation seems to be quite general and has now been observed in
a wide range of problems, e.g., see \cite{Kudekar-aller10,Hassani-itw10,Kudekar-isit11-DEC,Kudekar-isit11-MAC,Yedla-isit11,Yedla-aller11}%
\footnote{To be precise, the papers \cite{Kudekar-aller10,Kudekar-isit11-DEC,Kudekar-isit11-MAC}
only observe the threshold saturation effect indirectly because the
considered EXIT-like curves provide no direct information about the
MAP threshold of the underlying ensemble.%
}. 

In the realm of channels with memory and particularly intersymbol
interference (ISI) channels, the capacity may not be achievable via
equiprobable signaling. For linear codes, a popular practice is to
compare instead with the symmetric information rate (SIR), which is
also known as $C_{\text{i.u.d.}}$ \cite{Kavcic-it03}, because this
the rate is achievable by random linear codes with maximum-likelihood
(ML) decoding. A numerical method for tightly estimating the SIR of
finite-state channels in general was first proposed in \cite{Arnold-icc01,Pfister-globe01}.
For LDPC codes over ISI channels, a joint iterative BP decoder that
operates on a large graph representing both the channel and the code
constraints \cite{Kurkoski-it02,Kavcic-it03} can perform quite well
and even approach the SIR \cite{Pfister-brest03,Pfister-jsac08}.
Progress has been made on the design of SIR-approaching irregular
LDPC codes for some specific ISI channels \cite{Pfister-brest03,Varnica-comlett03,Narayanan-aller04,Soriaga-it07,Pfister-jsac08}.
However, channel parameters must be known at the transmitter for such
designs and therefore universality across ISI channels appears difficult
to achieve.

Since spatially-coupled codes and the threshold saturation effect
have now shown benefits in many communication problems, it is quite
natural to consider them as a potential candidate to \emph{universally}
approach the SIR of ISI channels with low decoding complexity. In
fact, the combination of spatially-coupled codes and ISI channels
was recently considered by Kudekar and Kasai \cite{Kudekar-isit11-DEC}
for the simple dicode erasure channel (DEC) from \cite{Pfister-03,Pfister-jsac08}.
They provided a numerical evidence that the joint BP threshold of
the spatially coupled codes can approach the SIR over the DEC (by
increasing the degrees while keeping the rate fixed). Also, they outlined
a tentative proof approach for the threshold saturation following
the ideas in \cite{Kudekar-it11}. However, the EXIT-like curves they
considered were not equipped with an area theorem and therefore could
not be directly connected with the MAP threshold of the underlying
ensemble. Thus, the threshold saturation effect was indirectly observed.

In this paper, we begin by considering the transmission of the spatially-coupled
codes over the class of generalized erasure channels (GECs) of which
the DEC and BEC are two particular examples. For these channels, we
provide a rigorous analysis of the upper bound on the MAP threshold
of LDPC codes by extending the analysis in \cite{Measson-it08} beyond
the BEC case %
\footnote{The upper bound technique on the MAP threshold for the DEC was first
mentioned in an earlier paper by one of the authors \cite{Wang-turbo08}%
}. For the DEC, we then employ a counting argument and present a numerical
evidence that this bound is indeed tight for regular ensembles. With
the MAP threshold determined, the threshold saturation phenomenon
can be observed to occur exactly for the several channels from the
GECs. Next, we also consider the case of more general ISI channels
where, by deriving the appropriate GEXIT curve and associated area
theorem, the MAP threshold upper bound can be computed and threshold
saturation can be seen. As a consequence, it is possible for spatially-coupled
codes to closely approach the SIR of ISI channels under joint iterative
BP decoding because regular LDPC codes can achieve the SIR under MAP
decoding \cite{Bae-jsac09}.

\section{Background}

In this section, we briefly describe our notation for ISI channels,
LDPC ensembles, the joint iterative decoder and spatially-coupled
codes.

\subsection{ISI Channels and the SIR}

Let the input alphabet $\mathcal{X}$ be finite, $\{X_{i}\}_{i\in\mathbb{Z}}$
be the discrete-time input sequence (i.e., $X_{i}\in\mathcal{X}$)
and $\{Y_{i}\}_{i\in\mathbb{Z}}$ be the discrete-time output sequence.
Many ISI channels of interest can be modeled by\begin{equation}
Y_{i}=\sum_{t=0}^{\nu}a_{t}X_{i-t}+N_{i},\label{eq:ISI}\end{equation}
where the channel memory is $\nu$, $\{a_{t}\}_{t=1}^{\nu}$ is the
set of tap coefficients and $\{N_{i}\}_{i\in\mathbb{Z}}$ is a sequence
of independent noise random variables. One can also write the above
as $Y_{i}=Z_{i}+N_{i}$ where $Z_{i}=\sum_{t=0}^{\nu}a_{t}X_{i-t}$
is the ISI output without noise. In this paper, we restrict ourself
to the class of binary-input ISI channels. Often, the tap coefficients
are represented through a transform domain polynomial $a(D)=\sum_{t=0}^{\nu}a_{t}D^{t}$.
For example, when $a(D)=1-D$, the channel is known as the dicode
channel. 

The main subject of Section \ref{sec:DEC} is the class of generalized
erasure channels (GECs) in \cite{Pfister-03,Pfister-jsac08}. For
the GEC, one can evaluate its SIR (see \cite{Pfister-03,Pfister-jsac08}
for details) as\begin{equation}
I_{s}(\epsilon)=1-\int_{0}^{1}f(t,\epsilon)\text{d}t\label{eq:SIRGEC}\end{equation}
where $f(t,\epsilon)$ is the function which maps $t$, the \emph{a
priori} erasure rate from the code, and the channel erasure rate $\epsilon$
to the erasure rate at the output of the channel detector \cite{Pfister-jsac08}.
Strictly speaking, in this paper we mainly consider a subclass of
the GECs where the channel output sequence can be modeled as a deterministic
mapping of the input sequence plus erasure noise.
\begin{example}
The simpliest example is the dicode erasure channel (DEC), which is
basically a $1$st-order differentiator (i.e., $a(D)=1-D$) whose
output is erased with probability $\epsilon$ and transmitted perfectly
with probability $1-\epsilon$. Furthermore, if the input bits are
differentially encoded prior to transmission, the resulting channel
is called the precoded dicode erasure channel (pDEC). The simplicity
of the channel models allows one to analyze the recursions used by
the Bahl-Cocke-Jelinek-Raviv (BCJR) algorithm \cite{Bahl-it74} to
compute \begin{equation}
f_{\text{DEC}}(t,\epsilon)=\frac{4\epsilon^{2}}{(2-t(1-\epsilon))^{2}}\label{eq:fDEC}\end{equation}
for the DEC and \begin{equation}
f_{\text{pDEC}}(t,\epsilon)=\frac{4\epsilon^{2}t(1-\epsilon(1-t))}{(1-\epsilon(1-2t))^{2}}\label{eq:fpDEC}\end{equation}
for the pDEC. For both cases, explicit calculations give $I_{s}=1-\frac{2\epsilon^{2}}{1+\epsilon}$
\cite{Pfister-jsac08}. Note that this formula also applies for the
BEC where one has $f(t,\epsilon)=\epsilon$ and $I_{s}(\epsilon)=1-\epsilon$.
\end{example}
Section \ref{sec:General-ISI} considers more general ISI channels
among which the most important is probably linear ISI channels with
additive white Gaussian noise (AWGN). For this class of ISI channels,
the SIR is given by%
\footnote{A vector $(X_{1},X_{2},\ldots,X_{n})$ is denoted by $X_{1}^{n}$
for convenience.%
}\[
C_{\text{i.u.d.}}=\lim_{n\rightarrow\infty}\frac{1}{n}I(X_{1}^{n};Y_{1}^{n})\Big{|}_{p_{X_{1}^{n}}(x_{1}^{n})=2^{-n}}.\]
Unfortunately, no closed-form solutions for the SIR are known in this
case. Instead, the numerical method described in \cite{Arnold-icc01,Pfister-globe01,Arnold-it06}
is typically used to give tight estimates of the SIR.

\subsection{\label{sub:jointBP}LDPC Ensembles and the Joint BP Decoder}

The standard irregular LDPC ensemble is characterized by its degree
distribution (d.d.), which represents the fraction of nodes (or edges)
of particular degrees. From the edge perspective, the d.d. pair consists
of two polynomials $\lambda(x)=\sum_{i\geq1}\lambda_{i}x^{i-1}$ and
$\rho(x)=\sum_{i\geq1}\rho_{i}x^{i-1}$ whose coefficients $\lambda_{i}$
(or $\rho_{i}$) give the fraction of edges that connect to bit (or
check) nodes of degree $i$. The LDPC ensemble can also be viewed
from the node perspective where its d.d. pair $L(x)=\sum_{i\geq1}L_{i}x^{i}$
and $R(x)=\sum_{i\geq1}R_{i}x^{i}$ have coefficients $L_{i}$ (or
$R_{i}$) equal to the fraction of bit (or check) nodes of degree
$i$. The design rate of an LDPC ensemble is given by \[
\mathtt{r}=1-\frac{L'(1)}{R'(1)}=1-\frac{\int_{0}^{1}\rho(x)\text{d}x}{\int_{0}^{1}\lambda(x)\text{d}x}.\]

When LDPC codes are transmitted over the ISI channels defined by (\ref{eq:ISI}),
one can construct a large graph by joining the code graph and the
channel graph together as depicted in Fig. \ref{fig:JointGraph}.
Working on this joint graph, a joint iterative decoder typically passes
the information back and forth between the channel detector and the
LDPC decoder. This technique is termed as turbo equalization and was
first considered by Douillard \emph{et al. }in the context of turbo
codes \cite{Douillard-ett95}. For analysis, we also require the addition
of a random scrambling vector to symmetrize the effective channel
\cite{Hou-it03}. This is very similar to using a random coset of
the LDPC code to allow analysis of the decoder using the all-zero
codeword assumption; this technique was also used in \cite{Kavcic-it03}
where they proved a concentration theorem and derived the density
evolution (DE) equations for ISI channels.%
\begin{figure}
\begin{centering}
\includegraphics[scale=0.6]{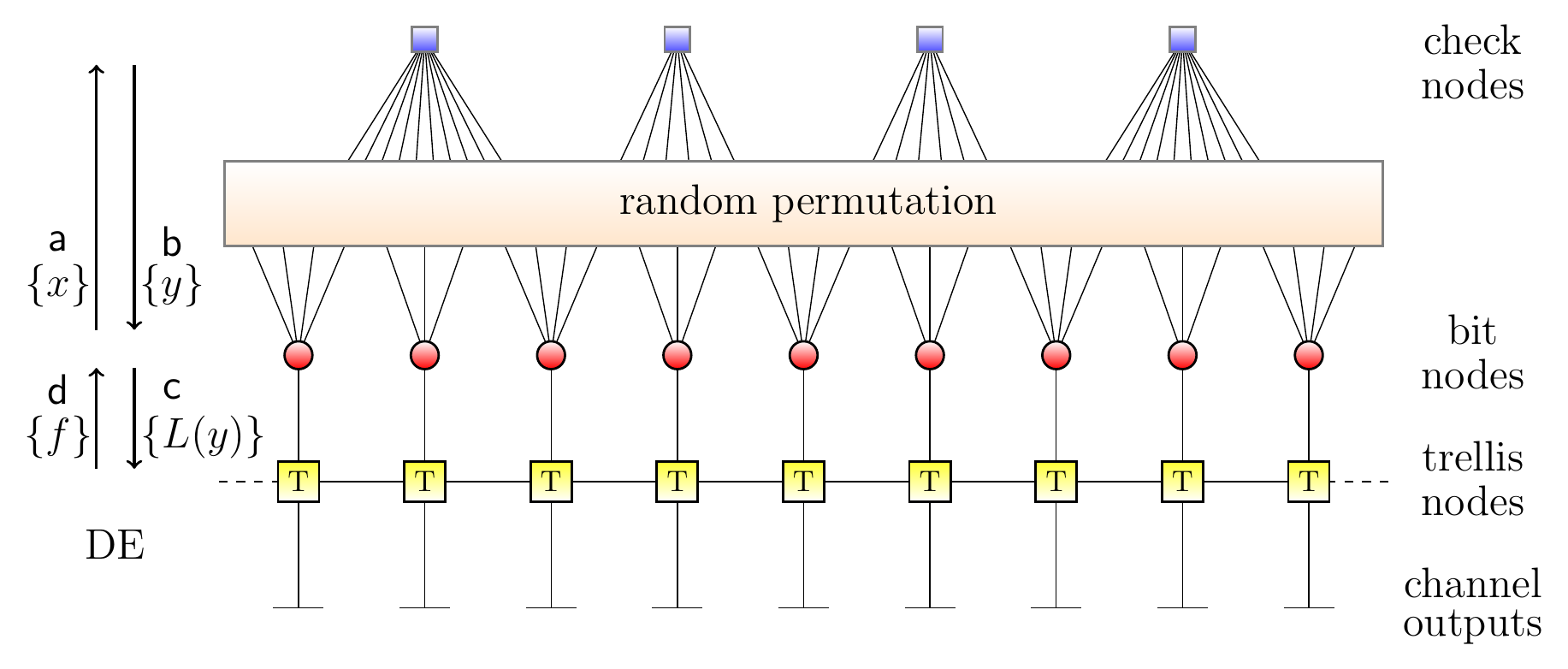}
\par\end{centering}

\caption{\label{fig:JointGraph}Gallager-Tanner-Wiberg graph of the joint BP
decoder for ISI channels. The notations $\mathsf{a},\mathsf{b},\mathsf{c},\mathsf{d}$
denote the average densities of the messages traversing along the
graph used in density evolution (DE). The quantities inside the brackets
are erasure rates used in DE for the GEC case. The update schedule
of the joint BP decoder is also implied by the arrows in this figure.}

\end{figure}

\subsection{Spatially-Coupled Ensembles}

The class of spatially-coupled ensembles in general can be defined
quite broadly. In this paper, we mainly consider two basic variants
(see details in \cite{Kudekar-it11}) as discussed below.

\subsubsection{The $(l,r,L)$ ensemble}

The $(l,r,L)$ spatially-coupled ensemble (with $l$ odd so that $\hat{l}=\frac{l-1}{2}\in\mathbb{N}$)
can be constructed from the underlying $(l,r)$-regular LDPC ensemble.
At each position from $[1,L]$ one has $M$ bit nodes and $\frac{l}{r}M$
check nodes just like in the $(l,r)$-regular case. However, each
bit node at position $i$ is connected to check nodes at the same
position, at $\hat{l}$ positions to the left and $\hat{l}$ positions
to the right (one check node from each position). In doing this, one
also needs to add $\frac{l}{r}M$ extra check nodes at each of $\hat{l}$
extra positions on each side. For example, a joint code/channel graph
for the $(3,6,L)$ ensemble and the ISI channels is shown in Fig.
\ref{fig:JGraph(3,6,L)}. The design rate of the $(l,r,L)$ ensemble
is given by\[
\mathtt{r}(l,r,L)=\left(1-\frac{l}{r}\right)-\frac{l}{r}\cdot\frac{l-1}{L}.\]

\begin{figure}
\begin{centering}
\includegraphics[scale=0.52]{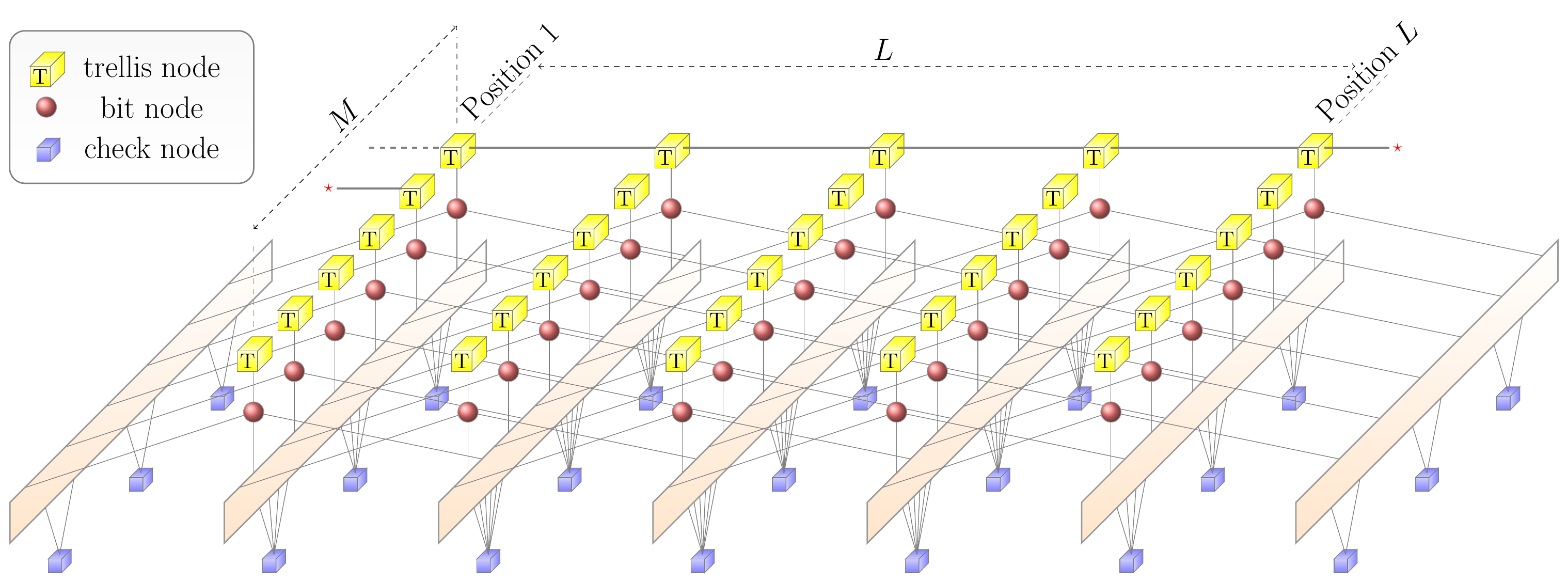}
\par\end{centering}

\caption{\label{fig:JGraph(3,6,L)}The joint graph for the $(l,r,L)$ ensemble
over the ISI channels. Illustrated in this figure is the case when
$l=3$ and $r=6$.}

\end{figure}

\subsubsection{The $(l,r,L,w)$ ensemble}

The $(l,r,L,w)$ can be obtained with the introduction of a {}``smoothing''
parameter $w$. One still places $M$ variable nodes at each position
in $[1,L]$ but places $\frac{l}{r}M$ check nodes at each position
in $[1,L+w\lyxmathsym{\textminus}1]$. Each bit node at position $i$
is connected uniformly and independently to a total of $l$ check
nodes at positions from the range $[i,i+w-1]$. By adding this randomization
of the edge connections with the parameter $w$, for large enough
$w$ the system behaves like a continuous one and a proof of the threshold
saturation effect becomes feasible \cite{Kudekar-it11}. The design
rate of the $(l,r,L,w)$ ensemble is given by\[
\mathtt{r}(l,r,L,w)=\left(1-\frac{l}{r}\right)-\frac{l}{r}\cdot\frac{w+1-2\sum_{i=0}^{w}\left(\frac{i}{w}\right)^{r}}{L}.\]

\section{\label{sec:DEC}ISI Channels with Erasure Noise: The GECs}

In this section, we focus on the class of GECs. We will present some
closed-form analyses on the (E)BP EXIT curves of the joint BP decoder.
This allows us to obtain an estimate of the MAP threshold of the underlying
ensemble. Then, DE is used to computed the BP thresholds of the corresponding
spatially-coupled ensembles and the threshold saturation effect is
demonstrated.

\subsection{BP and EBP Curves for the GEC}

For the class of GECs, the DE update equation of the joint BP decoder
is given by\[
x^{(\ell+1)}=f(L(1-\rho(1-x^{(\ell)}),\epsilon)\lambda(1-\rho(1-x^{(\ell)}))\]
where $x^{(\ell)}$ is the average erasure rate emitted from bit nodes
to check nodes during the $\ell$th iteration \cite{Pfister-jsac08}.

Let $x$ denote the limit of $x^{(\ell)}$ when $\ell\rightarrow\infty$.
The fixed point (FP) equation is then given by \begin{equation}
x=f(L(y(x)),\epsilon)\lambda(y(x))\label{eq:DEforGEC}\end{equation}
 where, for simplicity of notation, we use $y(x)\triangleq1-\rho(1-x)$
(and sometimes $y$ for short).

For most of the GECs, $f(t,\epsilon)$ is strictly increasing in $\epsilon$
for fixed $t$. In this case, there exists a unique function $\xi(t,v)$
such that $f(t,\xi(t,v))=v$ and one can obtain\begin{equation}
\epsilon(x)=\xi\left(L(y(x)),\frac{x}{\lambda(y(x))}\right).\label{eq:exeq}\end{equation}

\begin{example}
For the DEC case, one has $f(t,\epsilon)=\frac{4\epsilon^{2}}{(2-t(1-\epsilon))^{2}}$
and this gives the FP equation $x=\frac{4\epsilon^{2}\lambda(y)}{(2-L(y)(1-\epsilon))^{2}}.$
One can also solve for $\xi(t,v)=(2-t)/\left(\frac{2}{\sqrt{v}}-t\right)$
and gets \begin{equation}
\epsilon(x)=\frac{2-L(y(x))}{2\sqrt{\frac{\lambda(y(x))}{x}}-L(y(x))}.\label{eq:e(x)}\end{equation}
\end{example}
\begin{definitn}
\label{def:EXITf}Consider a d.d. $(\lambda,\rho)$ pair and the sequence
of LDPC ensembles LDPC($n,\lambda,\rho)$. For each $\mathcal{C}$
picked uniformly at random from $\text{LDPC}(n,\lambda,\rho)$, let
$X_{1}^{n}$ be chosen randomly and uniformly at random from $\mathcal{C}$
and and $Y_{1}^{n}$ be the received sequence after transmission over
a GEC with erasure rate $\epsilon$ and initial state $S_{0}$. The
associated EXIT function is defined as

\[
h(\epsilon)\triangleq\lim_{n\rightarrow\infty}\mathbb{E}_{\mathcal{C}}\left[\frac{\partial H(X_{1}^{n}|Y_{1}^{n}(\epsilon),S_{0})}{\partial\epsilon}\right].\]

\end{definitn}
When BP estimator is used at each bit instead of the optimal MAP estimator,
one also has the BP-EXIT function which is given by the following
definition.
\begin{definitn}
Consider the same setting as in Definition \ref{def:EXITf}, the associated
(joint) BP-EXIT function is defined to be\[
h^{\text{BP}}(\epsilon)\triangleq\lim_{\ell\rightarrow\infty}h^{\text{BP},\ell}(\epsilon)\]
where\[
h^{\text{BP},\ell}(\epsilon)=\lim_{n\rightarrow\infty}\mathbb{E}_{\mathcal{C}}\left[\frac{1}{n}\sum_{i=1}^{n}\frac{\partial H(X_{1}^{n}|Y_{i}(\epsilon),\mathcal{E}_{i}^{\text{BP,}\ell},S_{0})}{\partial\epsilon}\right]\]
and $\mathcal{E}_{i}^{\text{BP,}\ell}$ is the extrinsic BP estimate
of the $i$th bit after iteration $\ell$.\end{definitn}
\begin{lemma}
\label{lem:hvshbp}For simplicity of notation, let us write $Y_{\sim i}$
to denote the sequence $Y_{1}^{n}\setminus Y_{i}$. Then, the EXIT
function and BP-EXIT function (after iteration $\ell$) can be written
as\begin{align}
h(\epsilon) & =\lim_{n\rightarrow\infty}\mathbb{E}_{\mathcal{C}}\left[\frac{1}{n}\sum_{i=1}^{n}H(Z_{i}|Y_{\sim i}(\epsilon),S_{0})\right],\label{eq:hezi}\\
h^{\text{BP},\ell}(\epsilon) & =\lim_{n\rightarrow\infty}\mathbb{E}_{\mathcal{C}}\left[\frac{1}{n}\sum_{i=1}^{n}H(Z_{i}|\mathcal{E}_{i}^{\text{BP,}\ell}(Y_{\sim i}(\epsilon)),S_{0})\right].\label{eq:hebpzi}\end{align}
where $Z_{i}$ is the $i$th output without noise. From this, one
can see that $h(\epsilon)\leq h^{\text{BP}}(\epsilon)$.\end{lemma}
\begin{proof}
Let $\epsilon_{i}$ be the erasure rate of the channel from $Z_{i}$
to $Y_{i}$. For any extrinsic estimator $\mathcal{E}$, one has \begin{align*}
H(X_{1}^{n}|Y_{i}(\epsilon_{i}),\mathcal{E}(Y_{\sim i}),S_{0}) & =H(Z_{1}^{n}|Y_{i}(\epsilon_{i}),\mathcal{E}(Y_{\sim i}),S_{0})\\
 & =H(Z_{i}|Y_{i}(\epsilon_{i}),\mathcal{E}(Y_{\sim i}),S_{0})+H(Z_{\sim i}|Y_{i}(\epsilon_{i}),\mathcal{E}(Y_{\sim i}),Z_{i},S_{0})\\
 & =\epsilon_{i}H(Z_{i}|\mathcal{E}(Y_{\sim i}),S_{0})+H(Z_{\sim i}|Y_{i}(\epsilon_{i}),\mathcal{E}(Y_{\sim i}),Z_{i},S_{0}).\end{align*}

Since the second term on the R.H.S. does not depend on $\epsilon_{i}$,
it is clear that

\[
\frac{\partial H(X_{1}^{n}|Y_{i}(\epsilon_{i}),\mathcal{E}(Y_{\sim i}),S_{0})}{\partial\epsilon_{i}}=H(Z_{i}|\mathcal{E}(Y_{\sim i}),S_{0}).\]

By letting $\epsilon_{i}=\epsilon$ for all $i$ and considering two
specific cases of $\mathcal{E}$, one obtains (\ref{eq:hezi}) and
(\ref{eq:hebpzi}).

Furthermore, by data processing inequality \cite{Cover-1991}, one
has $H(Z_{i}|Y_{\sim i}(\epsilon),S_{0})\leq H(Z_{i}|\mathcal{E}_{i}^{\text{BP,}\ell}(Y_{\sim i}(\epsilon)),S_{0})$
which implies $h(\epsilon)\leq h^{\text{BP},\ell}(\epsilon)$ hence
$h(\epsilon)\leq h^{\text{BP}}(\epsilon)$.
\end{proof}
While computing the (MAP) EXIT function in general is hard, it is
relatively easy to compute the BP-EXIT function.
\begin{lemma}
The BP-EXIT function for the GEC is given by\begin{equation}
h^{\text{BP}}(\epsilon)=\frac{\partial}{\partial\tilde{\epsilon}}\int_{0}^{L(y)}f(t,\tilde{\epsilon})\text{d}t\Big|_{\tilde{\epsilon}=\epsilon}.\label{eq:hbpGEC}\end{equation}
where $L(y)$ is the extrinsic erasure rate given by the FP equation
at channel erasure rate $\epsilon$.\end{lemma}
\begin{proof}
Let $Y_{1}^{n}(\tilde{\epsilon})$ be the result of passing $X_{1}^{n}$
through the communication channel, e.g., the GEC, with erasure rate
$\tilde{\epsilon}$ and, with some abuse of notation, $\mathcal{E}_{1}^{n}(p)$
be the result of passing $X_{1}^{n}$ through the extrinsic channel
which is modeled as BEC with erasure probability $p.$ Similarly to
\cite{Pfister-jsac08}, let $T_{n}(1-t,\tilde{\epsilon})\triangleq\frac{1}{n}\sum_{i=1}^{n}I(X_{i};Y_{1}^{n}(\tilde{\epsilon}),\mathcal{E}_{\sim i}(p))$
denote the mutual information transfer function where $\mathcal{E}_{\sim i}$
comprises the sequence of extrinsic bit estimates except for the $i$th
bit. We also let $f_{n}(t,\tilde{\epsilon})\triangleq1-T_{n}(1-t,\tilde{\epsilon})$.
By the area theorem \cite[Th. 2]{Ashikhmin-it04,Measson-it08}, one
obtains\begin{equation}
\int_{0}^{\delta}\frac{1}{n}\sum_{i=1}^{n}H(X_{i}|Y_{1}^{n}(\epsilon),\mathcal{E}_{\sim i}(t))\text{d}t=\frac{1}{n}H(X_{1}^{n}|Y_{1}^{n}(\epsilon),\mathcal{E}_{1}^{n}(\delta)).\label{eq:area1}\end{equation}
We then have

\begin{align}
\frac{\partial}{\partial\tilde{\epsilon}}\int_{0}^{\delta}f_{n}(t,\tilde{\epsilon})\text{d}t & =-\frac{\partial}{\partial\tilde{\epsilon}}\int_{0}^{\delta}T_{n}(1-t,\tilde{\epsilon})\text{d}t\nonumber \\
 & =\frac{\partial}{\partial\tilde{\epsilon}}\int_{0}^{\delta}\left(-\frac{1}{n}\sum_{i=1}^{n}I(X_{i};Y_{1}^{n}(\tilde{\epsilon}),\mathcal{E}_{\sim i}(t))\right)\text{d}t\nonumber \\
 & =\frac{\partial}{\partial\tilde{\epsilon}}\int_{0}^{\delta}\left(\frac{1}{n}\sum_{i=1}^{n}H(X_{i})-\frac{1}{n}\sum_{i=1}^{n}I(X_{i};Y_{1}^{n}(\tilde{\epsilon}),\mathcal{E}_{\sim i}(t))\right)\text{d}t\label{eq:addconst}\\
 & =\frac{\partial}{\partial\tilde{\epsilon}}\int_{0}^{\delta}\frac{1}{n}\sum_{i=1}^{n}H(X_{i}|Y_{1}^{n}(\tilde{\epsilon}),\mathcal{E}_{\sim i}(t))\text{d}t\nonumber \\
 & =\frac{\partial}{\partial\tilde{\epsilon}}\left[\frac{1}{n}\sum_{i=1}^{n}H(X_{1}^{n}|Y_{1}^{n}(\tilde{\epsilon}),\mathcal{E}_{1}^{n}(\delta))\right]\label{eq:area}\end{align}
where (\ref{eq:addconst}) holds because $\frac{\delta}{n}\sum_{i=1}^{n}H(X_{i})$
is not a function of $\tilde{\epsilon}$ while (\ref{eq:area}) follows
from (\ref{eq:area1}).

If one considers the BP estimator, for each fixed $\ell$, by letting
$n\rightarrow\infty$, $f_{n}(t,\tilde{\epsilon})$ converges pointwise
to $f(t,\tilde{\epsilon})$ (see \cite{Pfister-jsac08}) while the
expectation of the R.H.S. of (\ref{eq:area}) converges to $h^{\text{BP},\ell}(\epsilon)$
if we choose $\epsilon=\tilde{\epsilon}$. Then, by letting $\ell\rightarrow\infty$,
one reaches a FP where $\delta\rightarrow L(y)$ and finally obtains\[
\frac{\partial}{\partial\tilde{\epsilon}}\int_{0}^{L(y)}f_{n}(t,\tilde{\epsilon})\text{d}t\Big|_{\tilde{\epsilon}=\epsilon}=h^{\text{BP}}(\epsilon).\]
 \end{proof}
\begin{example}
For the DEC and pDEC, using the result of (\ref{eq:fDEC}) and (\ref{eq:fpDEC}),
one has the following BP-EXIT functions\begin{align}
h_{\text{DEC}}^{\text{BP}}(\epsilon) & =\frac{\partial}{\partial\tilde{\epsilon}}\int_{0}^{L(y)}\frac{4\tilde{\epsilon}^{2}}{(2-t(1-\tilde{\epsilon}))^{2}}\text{d}t\Big|_{\tilde{\epsilon}=\epsilon}\nonumber \\
 & =\frac{2\epsilon L(y)(4-L(y)(2-\epsilon))}{(2-L(y)(1-\epsilon))^{2}}\label{eq:h(e(x))_}\end{align}
and

\begin{align*}
h_{\text{pDEC}}^{\text{EBP}}(\epsilon) & =\frac{\partial}{\partial\tilde{\epsilon}}\int_{0}^{L(y)}\frac{4\tilde{\epsilon}^{2}t(1-\tilde{\epsilon}(1-t))}{(1-\tilde{\epsilon}(1-2t))^{2}}\text{d}t\Big|_{\tilde{\epsilon}=\epsilon}\\
 & =\frac{2\epsilon L^{2}(y)(2-\epsilon(1-2L(y)))}{(1-\epsilon(1-2L(y)))^{2}}.\end{align*}
where $x$ is the DE FP at channel erasure rate $\epsilon$ and $y=y(x)$.
The formula (\ref{eq:h(e(x))_}) for the DEC case is equivalent to
the result shown in \cite{Wang-turbo08} by analyzing the BCJR algorithm.

Also, one can apply (\ref{eq:hbpGEC}) for the BEC to obtain a known
result $h_{\text{BEC}}^{\text{BP}}(\epsilon)=\frac{\partial}{\partial\tilde{\epsilon}}\int_{0}^{L(y)}\tilde{\epsilon}\text{d}t\Big|_{\tilde{\epsilon}=\epsilon}=L(y)$.
\end{example}
Using an approach similar to \cite[Sec. III-B]{Measson-it08} and
taking care of (\ref{eq:exeq}) and (\ref{eq:hbpGEC}), one gets the
following parametric form for the BP-EXIT function. This involves
in defining\[
\mathcal{I}\triangleq\bigcup_{i\in[J]}[\underline{x}^{i},\overline{x}^{i})\cup\{1\}\]
as the unique finite union of disjoint intervals that represent all
stable and achievable FPs of DE equations. Please note that $J$ represents
the number of discontinuties in the BP-EXIT function. For the case
$J\geq1$, let $x^{\text{BP}}=\underbar{x}^{1}$ and $\epsilon^{\text{BP}}=\epsilon(x^{\text{BP}})$
is the joint BP decoding threshold \cite[Sec. III-B]{Measson-it08}. 
\begin{lemma}
Given a d.d. pair $(\lambda,\rho)$, the BP-EXIT function for the
GEC is given parametrically as follows
\end{lemma}
\[
\left(\epsilon,h^{\text{BP}}(\epsilon)\right)=\begin{cases}
(\epsilon,0), & \epsilon\in[0,\epsilon^{\text{BP}})\\
\left(\epsilon(x),\frac{\partial}{\partial\tilde{\epsilon}}\int_{0}^{L(y(x))}f(t,\tilde{\epsilon})\text{d}t\Big|_{\tilde{\epsilon}=\epsilon(x)}\right)\,\,\forall x\in\mathcal{I}, & \epsilon\in(\epsilon^{\text{BP}},1]\end{cases}\]
where $\epsilon(x)$ is given in (\ref{eq:exeq}).

In \cite{Measson-it08}, the extended BP (EBP) EXIT curve for the
BEC was introduced as the hidden bridge between the BP threshold and
its MAP counterpart. In a similar manner, the EBP-EXIT curve for the
GEC is given below with its own area theorem.
\begin{definitn}
For a given d.d. pair $(\lambda,\rho)$, the EBP-EXIT curve for the
GEC is defined by the pair\[
\left(\epsilon(x),\frac{\partial}{\partial\tilde{\epsilon}}\int_{0}^{L(y(x))}f(t,\tilde{\epsilon})\text{d}t\Big|_{\tilde{\epsilon}=\epsilon(x)}\right),\, x\in[0,1]\]
where $\epsilon(x)$ is given in (\ref{eq:exeq}).\end{definitn}
\begin{example}
For the DEC case, using (\ref{eq:e(x)}) and (\ref{eq:h(e(x))_} ),
the EBP-EXIT curve is given by\[
\left(\frac{2-L(y(x))}{2\sqrt{\frac{\lambda(y(x))}{x}}-L(y(x))},L(y(x))\left(2\sqrt{\frac{x}{\lambda(y(x))}}-\frac{xL(y(x))}{2\lambda(y(x))}\right)\right),\, x\in[0,1].\]
\end{example}
\begin{lemma}
\label{lem:Trial Entropy}Consider the GEC and a d.d. pair $(\lambda,\rho)$.
Define the {}``trial entropy'' as \[
P(x)\triangleq\int_{0}^{x}h^{\text{EBP}}(t)\epsilon'(t)\text{d}t\]
 where $h^{\text{EBP}}(x)$ is the second coordinate the EBP-EXIT
curve. Then, we have\begin{align}
P(x) & =\int_{0}^{L(y)}f(t,\epsilon(x))\text{d}t-\frac{L'(1)}{R'(1)}(1-R(1-x)-xR'(1-x)).\label{eq:pxGEC}\end{align}

\end{lemma}

\begin{proof}
First, we let \begin{align}
Q(x) & \triangleq\int_{0}^{L(y)}f(t,\epsilon(x))\text{d}t-\frac{L'(1)}{R'(1)}(1-R(1-x)-xR'(1-x))\nonumber \\
 & =\int_{0}^{L(y)}f(t,\epsilon(x))\text{d}t-L'(1)\int_{0}^{x}u\text{d}y(u)\label{eq:Q1}\end{align}
where in (\ref{eq:Q1}), integration by parts is used.

Then, one can use Leibniz's rule to get

\begin{eqnarray}
Q'(x) & = & f(L(y),\epsilon(x))y'L'(y)+\int_{0}^{L(y)}\frac{\partial}{\partial x}f(t,\epsilon(x))\text{d}t-L'(1)xy'\nonumber \\
 & = & \int_{0}^{L(y)}\frac{\partial}{\partial x}f(t,\epsilon(x))\text{d}t\label{eq:Q'1}\\
 & = & \int_{0}^{L(y)}\frac{\partial}{\partial\epsilon(x)}f(t,\epsilon(x))\frac{\text{d}}{\text{d}x}\epsilon(x)\text{d}t\nonumber \\
 & = & \epsilon'(x)\int_{0}^{L(y)}\frac{\partial}{\partial\epsilon(x)}f(t,\epsilon(x))\text{d}t\nonumber \\
 & = & \epsilon'(x)h^{\text{EBP}}(x)\nonumber \\
 & = & P'(x)\label{eq:Q'2}\end{eqnarray}
where (\ref{eq:Q'1}) follows from the DE equation $f(L(y),\epsilon(x))\lambda(y)=x$
and the fact that $\lambda(y)=\frac{L'(y)}{L'(1)}$ while (\ref{eq:Q'2})
follows by taking deravative.

Thus, $Q(x)$ and $P(x)$ may differ by a constant. By seeing that
$P(0)=Q(0)=0$, one must have $P(x)\equiv Q(x)$.\end{proof}
\begin{example}
For the DEC, explicit calculation gives \begin{align*}
P(x) & =\frac{2\epsilon^{2}(x)L(y)}{2-L(y)(1-\epsilon(x))}-\frac{L'(1)}{R'(1)}(1-R(1-x)-xR'(1-x)).\end{align*}

Also, one can see that, for the BEC, this gives same formula as in
\cite{Measson-it08}.\end{example}
\begin{thm}
(Area Theorem for EBP) Consider a d.d. pair $(\lambda,\rho)$ of design
rate $\mathtt{r}$. Then the EBP EXIT curve for the GEC satisfies\[
\int_{0}^{1}h^{\text{EBP}}(x)\text{d}\epsilon(x)=\mathtt{r}.\]
\end{thm}
\begin{proof}
Using the result in Lemma \ref{lem:Trial Entropy}, a direct calculation
reveals that \[
\int_{0}^{1}h^{\text{EBP}}(x)\text{d}\epsilon(x)=P(1)=\int_{0}^{1}f(t,1)\text{d}t-\frac{L'(1)}{R'(1)}=1-\frac{L'(1)}{R'(1)}=\mathtt{r}\]
and the theorem is proven.
\end{proof}

\subsection{Upper Bound on the MAP Threshold}

Because of the optimality of the MAP decoder in the sense that $h^{\text{MAP}}\leq h^{\text{BP}}$
(see Lemma \ref{lem:hvshbp}), one can obtain an upper bound on the
MAP threshold by first finding the largest value $x^{\text{MAP}}$
such that $\int_{x^{\text{MAP}}}^{1}h^{\text{EBP}}(x)\text{d}\epsilon(x)=\mathtt{r}$
and then bound the MAP threshold by the inequality $\epsilon^{\text{MAP}}\leq\bar{\epsilon}^{\text{MAP}}\triangleq\epsilon(x^{\text{MAP}})$.
This technique was introduced by M{\'e}asson \emph{et al.} in \cite{Measson-it08}
in the context of BEC and conjectured to be tight in many scenarios.
In fact, for the whole class of regular LDPC ensembles over the BEC,
this bound was analytically proven to be tight \cite{Measson-isit07}. 

With the ingredients provided in our analysis above, the technique
can also be extended to the whole class of GECs. A corollary of Lemma
\ref{lem:Trial Entropy} implies in a few steps that one can find
$x^{\text{MAP}}$ as the unique solution of $P(x)=0$ in $[x^{\text{BP}},1]$.
From this, it is also clear that, $\bar{\text{\ensuremath{\epsilon}}}^{\text{MAP}}$
for the case of regular LDPC ensembles quickly approaches $\epsilon^{\text{SIR}}$
of the GEC which is formalized by the following theorem.
\begin{thm}
\label{thm:toSIR}Consider the $(l,r)$-regular ensemble. Consider
a fixed design rate $\mathtt{r}=1-\frac{l}{r}$. Then \[
\lim_{l,r\rightarrow\infty,\mathtt{r}\text{\,\ fixed}}\bar{\epsilon}^{\text{MAP}}(l,r)=\epsilon^{\text{SIR}}(\mathtt{r})\]
where $\epsilon^{\text{SIR}}(\mathtt{r})$ is the corresponding erasure
rate when SIR defined in (\ref{eq:SIRGEC}) equals $\mathtt{r}$.\end{thm}
\begin{proof}
First, $x^{\text{MAP}}(l,r)$ must be the solution of $P(x)=0$. For
a fixed rate $\mathtt{r}$, $x^{\text{MAP}}(l,r)$ is bounded away
from zero for $l$ large enough (one can show that $x^{\text{MAP}}(l,r)$
for the GEC is not less than $x_{\text{BEC}}^{\text{MAP}}(l,r)$ for
the BEC and the latter converges to $1-\mathtt{r}$ \cite[Lm. 8]{Kudekar-it11}).
Suppose that all the limits are taken when $l,r\rightarrow\infty$
while $\mathtt{r}$ is kept fixed. Then, we have $(1-x^{\text{MAP}}(l,r))^{r-1}\rightarrow0$
exponentially fast. 

Next, one also sees that \begin{equation}
L(y(x^{\text{MAP}}(l,r)))=(1-(1-x^{\text{MAP}}(l,r))^{r-1})^{l}\rightarrow1\,\,\text{and}\,\,\lambda(y(x^{\text{MAP}}(l,r)))\label{eq:Ly->1}\end{equation}
 which can be obtained from \begin{equation}
\frac{\log\left(1-(1-x^{\text{MAP}}(l,r))^{r-1}\right)}{1/(r-1)}\rightarrow0.\label{eq:lim1}\end{equation}
To see (\ref{eq:lim1}), we apply L'H\^opital's rule and use the
fact that \[
\frac{(1-x^{\text{MAP}}(l,r))^{r-1}}{\left(1-(1-x^{\text{MAP}}(l,r))^{r-1})\right)/(r-1)^{2}}\rightarrow0\]
because the numerator vanishes exponentially while the denominator
only vanishes quadratically fast.

Note that for $(l,r)$-regular ensemble, (\ref{eq:pxGEC}) can be
rewritten as\begin{equation}
P(x)=\int_{0}^{L(y)}f(t,\epsilon(x))\text{d}t+\frac{l}{r}(1-x)^{r-1}(1+(r-1)x)-\frac{l}{r}=0.\label{eq:regP(x)}\end{equation}
Therefore, we can use $P(x^{\text{MAP}}(l,r))=0$ and (\ref{eq:regP(x)}),
(\ref{eq:Ly->1}) to have \[
\int_{0}^{1}f(t,\epsilon(x^{\text{MAP}}(l,r)))\text{d}t\rightarrow\frac{l}{r}=1-\mathtt{r}.\]
In addition, from definition we have $\int_{0}^{1}f(t,\epsilon^{\text{SIR}}(\mathtt{r}))\text{d\ensuremath{t=1-I_{s}(\epsilon^{\text{SIR}}(\mathtt{r}))=1-\mathtt{r}}.}$
Therefore, \[
I_{s}(\bar{\epsilon}^{\text{MAP}}(l,r))\rightarrow I_{s}(\epsilon^{\text{SIR}}(\mathtt{r}))\]
and one has $\bar{\epsilon}^{\text{MAP}}(l,r)\rightarrow\epsilon^{\text{SIR}}(\mathtt{r})$
as $I_{s}(\cdot)$ is a continuous and monotone function.
\end{proof}

\begin{example}
Let us consider the DEC. For rate one-half ensembles, we have $\bar{\epsilon}^{\text{MAP}}(3,6)\approx0.638659$,
$\bar{\epsilon}^{\text{MAP}}(4,8)\approx0.640163$, $\bar{\epsilon}^{\text{MAP}}(5,10)\approx0.640355$,
$\bar{\epsilon}^{\text{MAP}}(7,14)\approx0.640387$, $\bar{\epsilon}^{\text{MAP}}(8,16)\approx0.640388$
that quickly approach $\epsilon^{\text{SIR}}(\frac{1}{2})\approx0.640388$.
This can be partially seen in Fig. \ref{fig:MAPvsSIR} where $\bar{\epsilon}^{\text{MAP}}(4,8)$
is already very close to $\epsilon^{\text{SIR}}.$%
\begin{figure}
\begin{centering}
\includegraphics[scale=0.65]{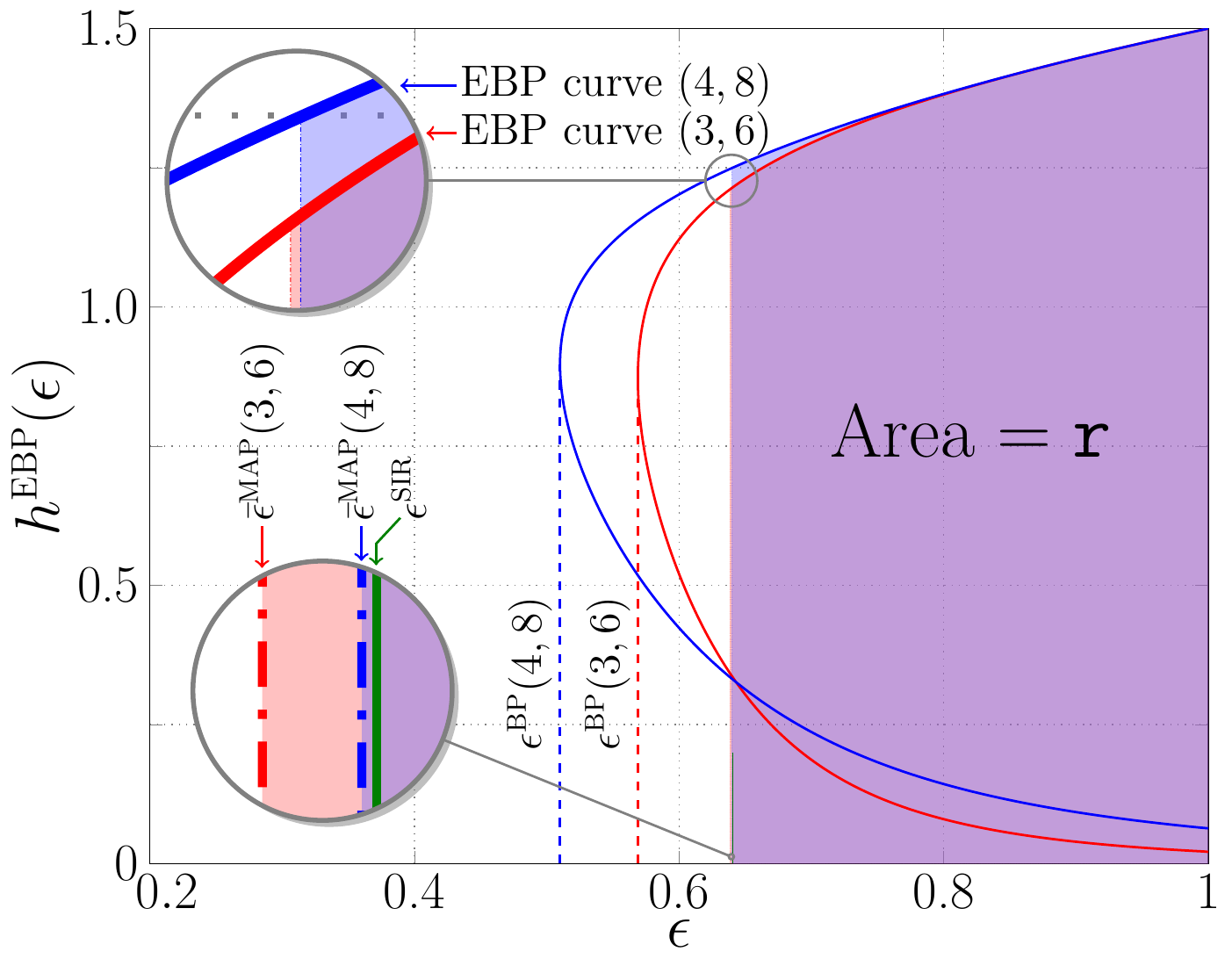}
\par\end{centering}

\caption{\label{fig:MAPvsSIR}EBP EXIT curves for $(3,6)$ and $(4,8)$ regular
LDPC ensembles over the DEC. Projection of the left most point of
the curves on to the $\epsilon$-axis allows one to determine $\epsilon^{\text{BP}}$.
Setting the area under the EBP curves to be equal to the design rate
$\mathtt{r}$ can help find $\bar{\epsilon}^{\text{MAP}}$.}

\end{figure}

\end{example}

\subsection{Tightness of the upper bound}

In this section, we discuss the tightness of the $\bar{\epsilon}^{\text{MAP}}$
bounding technique. Assume that the joint BP decoder is run on the
joint graph of the LDPC code and GEC. Since one never gets errors
in the GEC, the joint BP decoder must reach a FP where no more bit
nodes can be decoded. At this FP, one obtains a residual graph (see
\cite[Ch. 3]{RU-2008}) by removing all the known bit nodes as well
as their neighboring check nodes and the edges connecting them. Then,
one can follow the general procedure to show that the MAP bounding
technique is tight, i.e., by seeing at channel erasure rate $\bar{\epsilon}^{\text{MAP}}$,
the design rate of the residual graph is zero and providing numerical
evidence that for this residual graph, the actual rate converges to
the design rate as the blocklength $n\rightarrow\infty$. 

We start with the following lemma.
\begin{lemma}
\label{lem:ResidualGEC}Consider a d.d. pair $(\lambda,\rho)$ and
the GEC with channel erasure rate $\epsilon$. First, run the joint
BP decoder until it reaches a FP so that we obtain a residual graph.
Next, use the remaining channel constraints to merge all bit nodes
that must have the same value. The expected check node d.d. of the
residual graph%
\footnote{The check node and bit node d.d. are normalized with respect to the
original graph.%
} is given by\begin{equation}
\tilde{R}_{\epsilon}(z)=R(1-x+zx)-R(1-x)-zxR'(1-x)\label{eq:RzGEC}\end{equation}
where $x$ is the FP of DE and $y=1-\rho(1-x)$. Furthermore, if the
expected bit node d.d. is \begin{equation}
\tilde{L}_{\epsilon}(z)=\int_{0}^{L(yz)}f(t,\epsilon)\text{d}t\label{eq:LzGEC}\end{equation}
then at $\epsilon=\bar{\epsilon}^{\text{MAP}}$, the design rate of
the residual graph $\tilde{\mathtt{r}}_{\bar{\epsilon}^{\text{MAP}}}$
equals zero.\end{lemma}
\begin{proof}
Consider the original graph at the FP and let $x$ be the average
erasure rate from a bit node to a check node. Pick a check node of
degree $j$ in the original graph. We can obtain a check node of degree
$i\leq j$ in the residual graph by removing all $(j-i)$ edges with
known values. Note that $i\geq2$ since a check node of degree one
must not be in the residual graph. The remaining $i$ edges of this
check node must contain erasure messages. The probability for this
event is ${j \choose i}(1-x)^{(j-i)}x^{i}$. Thus, the check node
d.d. for the residual graph (normalized by the number of check nodes
in the original graph) is%
\footnote{This formula is the same as the check node d.d. for residual graph
left by the peeling decoder for the BEC, obtained via solving a differential
equation in \cite{Luby-it01}.%
}\begin{align*}
\tilde{R}_{\epsilon}(z) & =\sum_{j\geq2}R_{j}\sum_{i=2}^{j}{j \choose i}(1-x)^{(j-i)}(xz)^{i}\\
 & =R(1-x+zx)-R(1-x)-zxR'(1-x)\end{align*}
and (\ref{eq:RzGEC}) holds.

Suppose the bit node d.d. is given by (\ref{eq:LzGEC}), one has $\tilde{L}'_{\epsilon}(z)=y'L'(yz)f(L(yz),\epsilon)$
and $\tilde{R}'_{\epsilon}(z)=xR'(1-x+zx)-xR'(1-x).$ Therefore, one
obtains\begin{align}
\frac{\tilde{L}'_{\epsilon}(1)}{\tilde{R}'_{\epsilon}(1)} & =\frac{yL'(y)f(L(y),\epsilon)}{xR'(1)(1-\rho(1-x))}\nonumber \\
 & =\frac{L'(1)}{R'(1)}\cdot\frac{\lambda(y)f(L(y),\epsilon)}{x}\nonumber \\
 & =\frac{L'(1)}{R'(1)}\label{eq:LoverR1}\end{align}
by using (\ref{eq:DEforGEC}), $y=1-\rho(1-x)$ and the known facts
that $L'(y)=L'(1)\lambda(y)$ and $R'(1-x)=R'(1)\rho(1-x)$. 

Note that the standard d.d. pair from the node perspective of the
residual graph is $\left(\frac{\tilde{L}_{\epsilon}(z)}{\tilde{L}_{\epsilon}(1)},\frac{\tilde{R}_{\epsilon}(z)}{\tilde{R}_{\epsilon}(1)}\right)$
and the corresponding design rate is then \[
\tilde{\mathtt{r}}_{\epsilon}=1-\frac{\tilde{L}_{\epsilon}'(1)}{\tilde{R}_{\epsilon}'(1)}\cdot\frac{\tilde{R}_{\epsilon}(1)}{\tilde{L}_{\epsilon}(1)}.\]

Using (\ref{eq:LoverR1}), it now is clear that\[
\tilde{\mathtt{r}}_{\epsilon}=1-\frac{L'(1)}{R'(1)}\cdot\frac{\tilde{R}_{\epsilon}(1)}{\tilde{L}_{\epsilon}(1)}=\frac{P(x)}{\tilde{L}_{\epsilon}(1)}\]
where the last equality follows from (\ref{eq:RzGEC}), (\ref{eq:LzGEC})
and (\ref{eq:pxGEC}).

By considering a special case $\epsilon=\bar{\epsilon}^{\text{MAP}}$,
one has $\tilde{\mathtt{r}}_{\bar{\epsilon}^{\text{MAP}}}=P(x^{\text{MAP}})/\tilde{L}_{\bar{\epsilon}^{\text{MAP}}}(1)=0.$\end{proof}
\begin{remrk}
For the BEC, the bit node d.d. given in (\ref{eq:LzGEC}) matches
the known result in \cite[Th. 3.106]{RU-2008}. In fact, this also
holds for the DEC case which can be shown by the following lemma.\end{remrk}
\begin{lemma}
\label{lem:Residual d.d.-1}Consider a d.d. pair $(\lambda,\rho)$
and the DEC with erasure probability $\epsilon$. The expected bit
node d.d. in this case follows the form (\ref{eq:LzGEC}), i.e., \begin{equation}
\tilde{L}_{\epsilon}(z)=\frac{2\epsilon^{2}L(yz)}{2-L(yz)(1-\epsilon)}=\sum_{k=0}^{\infty}\epsilon^{2}\left(\frac{1-\epsilon}{2}\right)^{k}L(yz)^{k+1}\label{eq:Lz}\end{equation}
Consequently, at $\epsilon=\bar{\epsilon}^{\text{MAP}}$ the design
rate of the residual graph equals zero.\end{lemma}
\begin{proof}
The bit nodes in the residual graph must connect to the trellis section
of the form depicted in Fig. \ref{fig:Residual Trellis-1} for some
$k\in\mathbb{N}$ (otherwise, the joint BP decoder can still decode).

\begin{figure}[H]
\begin{centering}
\includegraphics[scale=0.75]{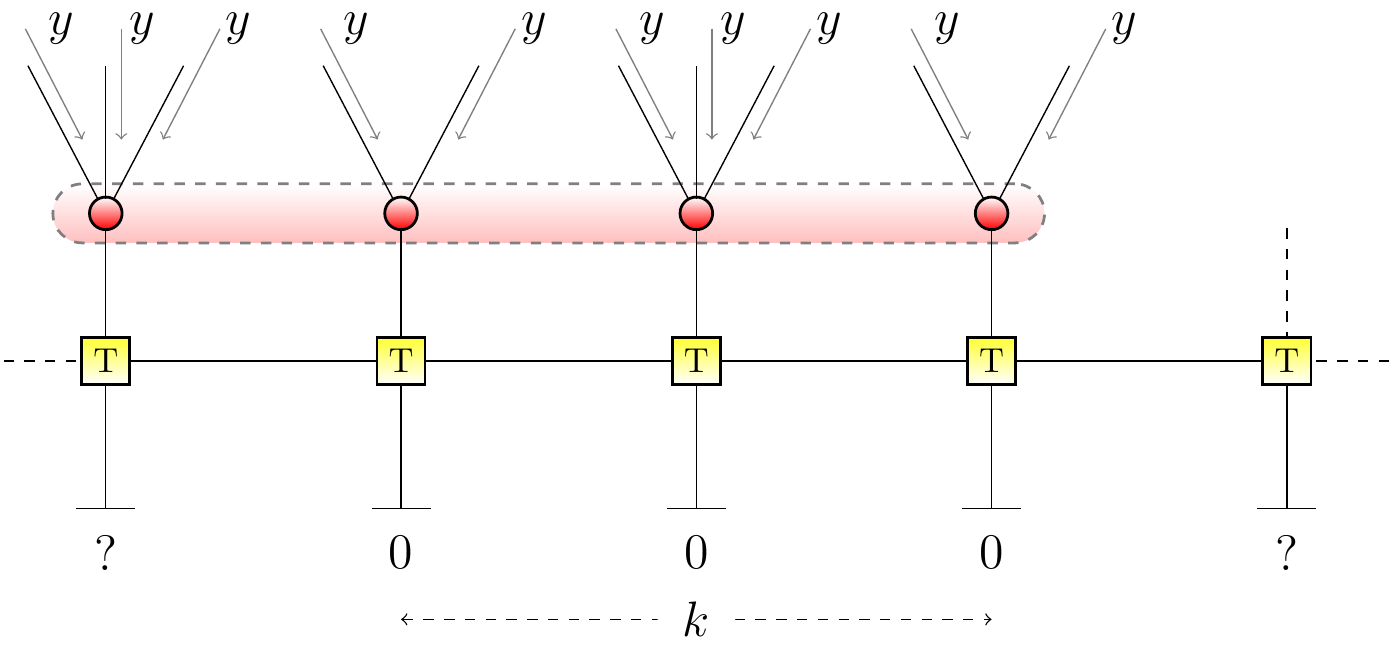}
\par\end{centering}

\caption{\label{fig:Residual Trellis-1}A trellis section in the residual graph
for the DEC. The notation {}``$?$'' denotes that an erasure is
received at the channel output. One can form a larger bit node by
merging all the bit nodes that attach to this trellis section.}

\end{figure}
The probability of the trellis configuration $(?,\overbrace{0,\ldots,0}^{k},?)$,
where $?$ indicates an erasure, is $\epsilon^{2}\left(\frac{1-\epsilon}{2}\right)^{k}$.
Given the above trellis configuration, if all messages from check
nodes to the bit nodes that attach to this trellis section are {}``$?$''
then \emph{all} these bit nodes remain in the residual graph. On the
other hand, if at least one of the messages is not {}``$?$'', then
the joint BP decoder can decode and then remove \emph{all} these bit
nodes from the residual graph. Therefore, one can consider all the
bit nodes that attach to such a trellis section as one larger bit
node whose degree is the sum of the $k+1$ component degrees. The
generating function for this sum of $k+1$ i.i.d. random variables
is $L(z)^{k+1}$. This is quite similar to the graph reduction technique
discussed in \cite{Pfister-it07} for IRA/ARA codes.

From the above analysis and since each edge is associated with erasure
rate $y$, the d.d. (normalized by the number of bit nodes in the
original graph) of residual graph after graph reduction is then given
by\[
\tilde{L}_{\epsilon}(z)=\sum_{k=0}^{\infty}\epsilon^{2}\left(\frac{1-\epsilon}{2}\right)^{k}L(yz)^{k+1}=\frac{2\epsilon^{2}L(yz)}{2-L(yz)(1-\epsilon)}.\]

\end{proof}
From the above analysis, once one has $\tilde{\mathtt{r}}_{\bar{\epsilon}^{\text{MAP}}}=0$,
the final missing piece to prove the tightness of the MAP upper bound
is to show that the actual rate of the residual graph is equal to
its design rate with high probability (when the blocklength tends
to $\infty$)%
\footnote{If this is true, then the MAP decoder can decode perfectly at $\bar{\epsilon}^{\text{MAP}}$
and $\bar{\epsilon}^{\text{MAP}}=\epsilon^{\text{MAP}}$.%
}. While a general proof for this still requires some analytic work,
one can use the test in \cite[Lm. 3.22]{RU-2008} to numerically verify
if this is true. To do this, one just needs to show that the function
$\Psi(u)$ introduced in \cite[Lm. 3.22]{RU-2008}, for the residual
graph, has the following property: $\Psi(u)\leq0$ in the interval
$[0,1]$ with equality only at $u=0$ and $u=1$. For our case, the
bit node d.d. for the residual graph from (\ref{eq:LzGEC}) might
have unbounded degrees as in (\ref{eq:Lz}) for the DEC case. However,
for this DEC case, since the fraction of bit nodes, for the $(l,r)$-regular
ensemble, that have degree $l(k+1)$ is upper bounded by $(\frac{1}{2})^{k}$
and therefore $\tilde{L}_{\epsilon}(z)$ has an exponentially vanishing
tail, one can truncate the series $\tilde{L}_{\epsilon}(z)$ at some
large enough $k$ and obtain the result with a negligible error. For
example, one can truncate $\tilde{L}_{\epsilon}(z)$ at $k=20$ and
for the $(3,6)$-regular ensemble, the truncated version of $\Psi(u)$
is numerically shown to satisfy the desired property in Fig. \ref{fig:Psi(3,6)}.%
\begin{figure}
\begin{centering}
\includegraphics[scale=0.65]{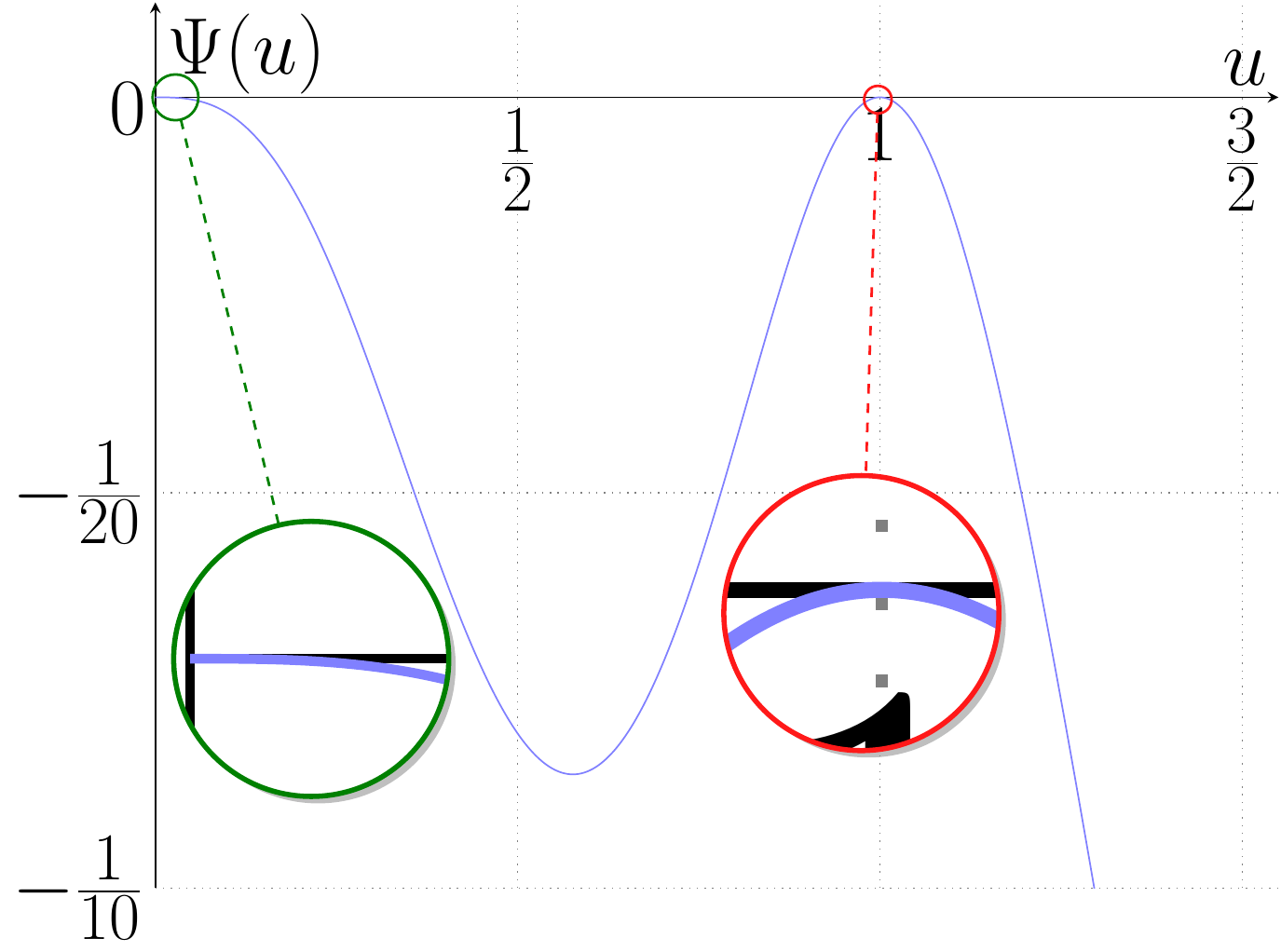}
\par\end{centering}

\caption{\label{fig:Psi(3,6)}Function $\Psi(u)$ for the residual graph obtained
after joint BP decoding of the $(3,6)$-regular LDPC ensemble over
the DEC. This shows numerically that the MAP upper bound is tight
in this case. }

\end{figure}

\subsection{Spatially-Coupled Codes for the GEC}

Consider the $(l,r,L,w$) spatially-coupled ensemble over the GEC.
The joint code/channel graph is similar to the one in Fig. \ref{fig:JGraph(3,6,L)}
which is for the $(l,r,L)$ ensemble. We also follow the DE equation
discussed in \cite{Kudekar-isit11-DEC} to compute the BP thresholds
of the coupled ensembles. The main difference is that we use the correct
EBP curves with their operational meaning instead of the EXIT-like
ones used in \cite{Kudekar-isit11-DEC}. Let $x_{i}^{(\ell)}$ denote
the expected erasure rate at iteration $\ell$ from bit nodes at position
$i$ to check nodes. For $i\notin[1,L]$, set $x_{i}^{(\ell)}=0$.
Let us define\begin{align*}
g(x_{i-w+1},\ldots,x_{i+w-1}) & \triangleq\left(1-\frac{1}{w}\sum_{j=0}^{w-1}\left(1-\frac{1}{w}\sum_{k=0}^{w-1}x_{i+j-k}\right)^{r-1}\right)^{l-1},\\
\Gamma(x_{i-w+1},\ldots,x_{i+w-1}) & \triangleq\left(1-\frac{1}{w}\sum_{j=0}^{w-1}\left(1-\frac{1}{w}\sum_{k=0}^{w-1}x_{i+j-k}\right)^{r-1}\right)^{l}.\end{align*}
 The DE equation for the joint BP decoder can be written as\begin{align*}
x_{i}^{(\ell+1)} & =f(\Gamma(x_{i-w+1}^{(\ell)},\ldots,x_{i+w-1}^{(\ell)}),\epsilon)\cdot g(x_{i-w+1}^{(\ell)},\ldots,x_{i+w-1}^{(\ell)})\end{align*}
for $i\in[1,L]$. To compute both the stable and unstable FPs of
DE, one can use the fixed entropy DE procedure outlined in \cite[Sec. VIII]{Measson-it09}
where the normalized entropy of a constellation $\underline{x}^{(\ell)}=(x_{1}^{(\ell)},\ldots,x_{L}^{(\ell)})$,
which is defined as $\chi(\underline{x}^{(\ell)})=\frac{1}{L}\sum_{i=1}^{L}x_{i}^{(\ell)}$,
is kept constant at every iteration by varying the channel parameter.
With each FP $\underline{x}$ obtained, one can compute the EBP EXIT
value of the spatially-coupled ensemble as $\frac{1}{L}\sum_{i=1}^{L}h^{\text{EBP}}(x_{i})$.

\begin{figure}
\begin{centering}
\includegraphics[scale=0.65]{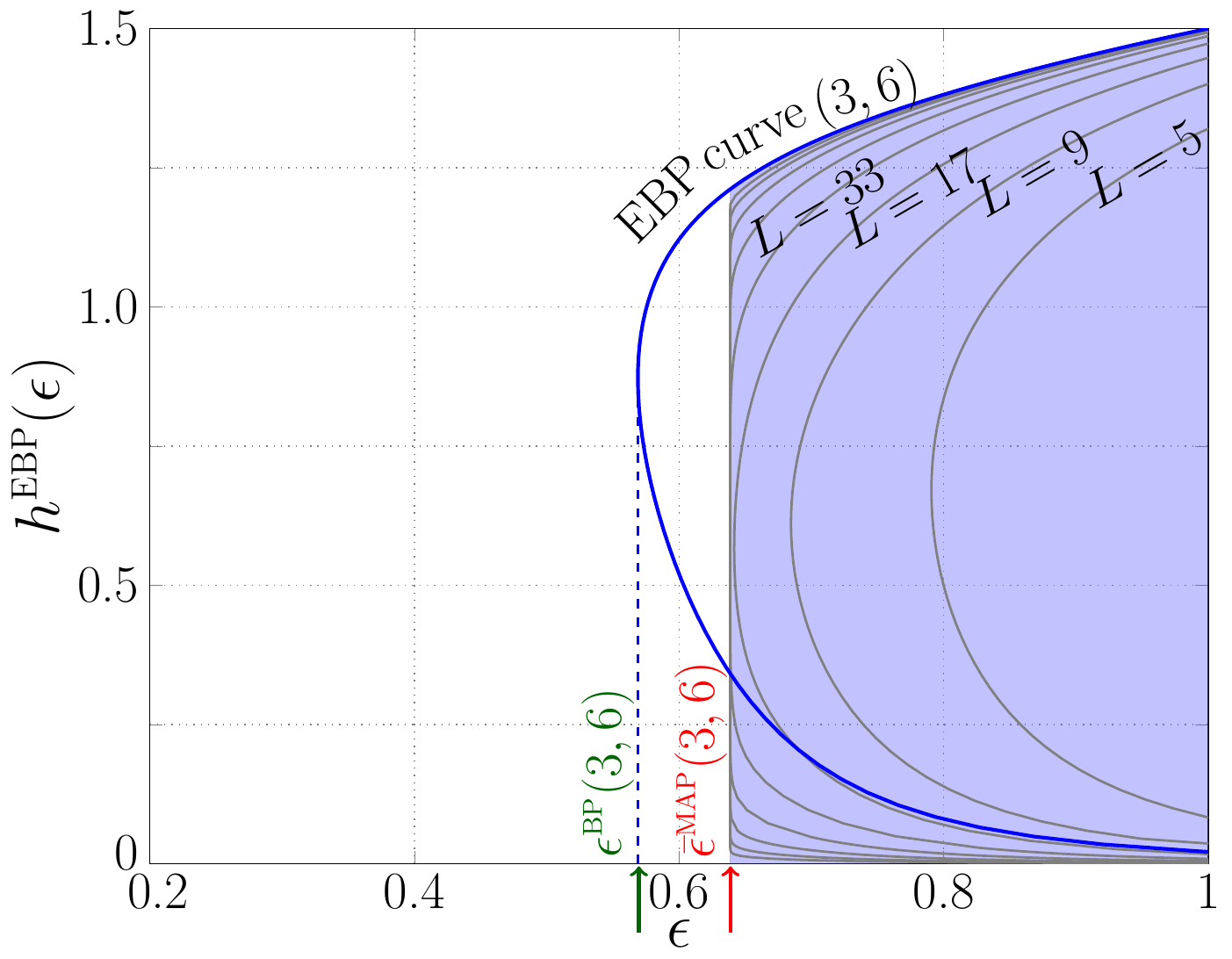}
\par\end{centering}

\caption{\label{fig:EBP-EXIT(3,6,L,5)}EBP EXIT curves for $(3,6,L,5)$ over
the DEC with $L=2\hat{L}+1$ where $\hat{L}=2,4,8,16,32,64,128,246$.
For small values of $L$, the increase in threshold can be explained
by the large rate-loss. As $L$ grows larger, the rate loss becomes
negligible and the curves keep moving left, but they saturate at the
MAP threshold of the underlying regular ensemble.}

\end{figure}

\begin{figure}
\begin{centering}
\includegraphics[scale=0.65]{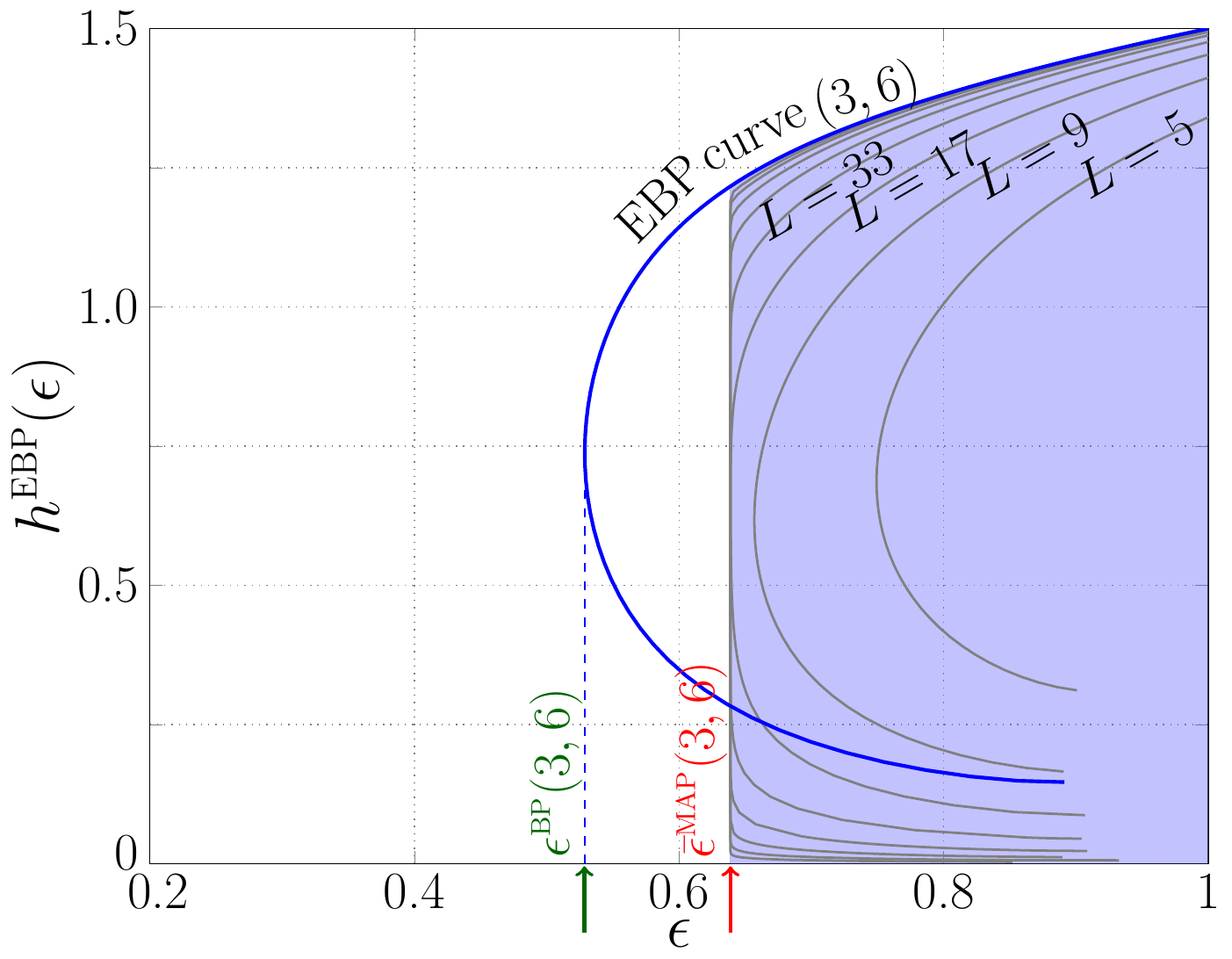}
\par\end{centering}

\caption{\label{fig:EBP-EXIT(3,6,L,5)-1}EBP EXIT curves for $(3,6,L,5)$ over
the pDEC with $L=2\hat{L}+1$ where $\hat{L}=2,4,8,16,32,64,128,246$.
Threshold saturation can also be observed for this case.}

\end{figure}

The threshold saturation effect of coupling can be nicely seen by
plotting the EBP EXIT curves for the uncoupled and coupled codes.
For the DEC, Fig. \ref{fig:EBP-EXIT(3,6,L,5)} shows the EBP curves
for the $(3,6,L,5)$ ensembles with various $L$ along with the EBP
curve of the underlying $(3,6)$-regular ensemble. From the EBP curves,
one can determine $\epsilon^{\text{BP}}(3,6)\approx0.56892$ and $\bar{\epsilon}^{\text{MAP}}(3,6)\approx0.63866$.
The BP thresholds of spatially-coupled ensembles for small $L$ due
to rate-loss can have larger values, e.g., $\epsilon^{\text{BP}}(3,6,17,6)\approx0.64170>\bar{\epsilon}^{\text{MAP}}(3,6)$.
However, for a wide range of $L$, i.e., $L=33,65,129,257,513$, we
observe that $\epsilon^{\text{BP}}(3,6,L,5)\approx0.63866$ which
is essentially $\bar{\epsilon}^{\text{MAP}}(3,6)$ while the rate
loss gradually becomes insignificant. In \cite{Kudekar-isit11-DEC}
, Kudekar and Kasai provided a similar plot but here we include the
MAP threshold estimate $\bar{\epsilon}^{\text{MAP}}$ and use the
EXIT function $h^{\text{EBP}}$ instead of the EXIT-like $L(y)$ in
\cite{Kudekar-isit11-DEC}. Similarly, one can also verify the threshold
saturation over the pDEC channel as seen in Fig. \ref{fig:EBP-EXIT(3,6,L,5)-1}.
For the pDEC, the\textbf{ }BP threshold for $(3,6)$-regular ensemble
is $\epsilon^{\text{BP}}(3,6)\approx0.52877$ and by using spatial
coupling, the BP threshold can be boosted to $\epsilon^{\text{BP}}(3,6,L,5)\approx\bar{\epsilon}^{\text{MAP}}(3,6)\approx0.63877$
with negligible rate loss for $L$ large. 

Even though the threshold saturation effect has been only shown numerically
for the DEC and pDEC, the method is readily applicable to the whole
class of GECs. Still, the analytic proof for threshold saturation
remains open for the GEC. Such a proof combining with Theorem \ref{thm:toSIR}
would essentially demonstrate the SIR-achieving capability of spatially-coupled
ensembles.

\section{\label{sec:General-ISI}General ISI Channels}

In this section, we shift our focus to ISI channels with more general
noise models. The MAP upper bound for general BMS channels was presented
by M{\'e}asson \emph{et al.} and conjectured to be tight \cite{Measson-it09}.
For general ISI channels, we apply a similar technique to give an
estimate of the MAP threshold of the underlying uncoupled ensemble
by first constructing the BP-GEXIT curve that follows an area theorem.
While our method can be used for a wide range of noise models, we
particularly focus on the case of AWGN. The BP thresholds of the corresponding
coupled ensembles are then computed via DE and the threshold saturation
effect is also observed. In addition, simulations on the performance
of the joint BP decoder for coupled codes of finite length are conducted
to validate these thresholds.

\subsection{\label{sub:GEXIT-Curves}GEXIT Curves for the ISI channels}

Consider an ISI channel of memory $\nu$. When the channel input $X_{1}^{n}$
is chosen uniformly at random from a suitable binary linear code%
\footnote{The code is proper \cite[p. 14]{RU-2008} and its dual code contains
no codewords involving only $0$'s and a run of $(\nu+1)$ $1$'s.%
}, the ISI output without noise $Z_{i}$ at some index $i$ is a discrete
random variable characterized by its probability mass function $p_{Z_{i}}(z)$
for all values $z$ in the alphabet $\mathcal{Z}$. For example, in
the case of a dicode channel, $\mathcal{Z}=\{0,+2,-2\}$ and $p_{Z_{i}}(0)=\frac{1}{2},p_{Z_{i}}(+2)=p_{Z_{i}}(-2)=\frac{1}{4}$.
The channel from $Z_{i}$ to $Y_{i}$ is a $|\mathcal{Z}|$-ary input
memoryless channel characterized by its transition probability density
$p_{Y_{i}|Z_{i}}(y|z)$. Without specifying the index, we denote $\mathtt{h}\triangleq H(Z|Y)$
and get\begin{align*}
\mathtt{h} & =H(Z)-I(Z;Y)\\
 & =H(Z)-\int_{-\infty}^{\infty}\sum_{z}p(z)p(y|z)\log_{2}\left\{ \frac{p(y|z)}{\sum_{z'}p(z')p(y|z')}\right\} \text{d}y.\end{align*}

Instead of looking at a particular channel, we assume that the channel
from $Z_{i}$ to $Y_{i}$ is from a smooth family $\{M(\mathtt{h}_{i})\}_{\mathtt{h}_{i}}$
of $|\mathcal{Z}|$-ary input memoryless channels characterized by
conditional entropy $\mathtt{h}_{i}$. A further assumption is made
that all individual channel families are parameterized in a smooth
way by a common parameter%
\footnote{For AWGN case, a convenient choice for $\epsilon$ is $\epsilon=-\frac{1}{2\sigma^{2}}$. %
} $\epsilon$, i.e., $\mathtt{h}_{i}=H(Z_{i}|Y_{i})(\epsilon)$.

With the convention that $y_{\sim i}\triangleq y_{1}^{n}\setminus y_{i}$,
define $\phi_{i}(y_{\sim i})\triangleq\left\{ P_{Z_{i}|Y_{\sim i}}(z|y_{\sim i}):z\in\mathcal{Z}\right\} $
and the random vector $\Phi_{i}\triangleq\phi_{i}(Y_{\sim i})$. Each
value of $\phi_{i}$ is a vector of length $|\mathcal{Z}|$ in the
$(|\mathcal{Z}|-1)$-dimensional probability simplex. The index of
the vector associated with $z\in\mathcal{Z}$ is denoted by $[z]$.
One can see that $\Phi_{i}$ is a sufficient statistic for estimating
$Z_{i}$, i.e., $Z_{i}\rightarrow\Phi_{i}(Y_{\sim i})\rightarrow Y_{\sim i}$
forms a Markov chain%
\footnote{One way to see this is to write \[
P_{Y_{\sim i}|Z_{i}}(y_{\sim i}|z_{i})=\frac{P_{Z_{i}|Y_{\sim i}}(z_{i}|y_{\sim i})}{P_{Z_{i}}(z_{i})}P_{Y_{\sim i}}(y_{\sim i})=\frac{\Phi_{i}\cdot e_{[z_{i}]}^{T}}{P_{Z_{i}}(z_{i})}P_{Y_{\sim i}}(y_{\sim i}),\]
where $e_{[z]}^{T}$ is the standard basis column vector with a 1
in the index $[z]$, and apply the result from \cite[p. 29]{RU-2008}.%
}.
\begin{definitn}
Suppose the initial state in the trellis is $S_{0}$. Let $X_{1}^{n}$
chosen according to $p_{X_{1}^{n}}(x_{1}^{n})$ be the input sequence,
$Z_{1}^{n}$ be the ISI output sequence without noise and $Y_{1}^{n}$
be the final channel output sequence, i.e., $Y_{i}$ is the result
of transmitting $Z_{i}$ over the smooth family $\{M(\mathtt{h}_{i})\}_{\mathtt{h}_{i}}$
of memoryless channels. Then the $i$th GEXIT function is\begin{equation}
\mathsf{G}_{i}(\mathtt{h}_{1},\ldots,\mathtt{h}_{n})=\frac{\partial H(X_{1}^{n}|Y_{1}^{n}(\mathtt{h}_{1},\ldots,\mathtt{h}_{n}),S_{0})}{\partial\mathtt{h}_{i}}\label{eq:gih1hn}\end{equation}
and the average GEXIT function is defined by \[
\mathsf{G}(\mathtt{h}_{1},\ldots,\mathtt{h}_{n})=\frac{1}{n}\sum_{i=1}^{n}\mathsf{G}_{i}(\mathtt{h}_{1},\ldots,\mathtt{h}_{n}).\]
 For the case where all channel families are the same, i.e., $\mathtt{h}_{i}=\mathtt{h}$,
we have \[
\mathsf{G}(\mathtt{h})=\frac{1}{n}\cdot\frac{\text{d}H(X_{1}^{n}|Y_{1}^{n}(\mathtt{h}),S_{0})}{\text{d}\mathtt{h}}.\]
\end{definitn}
\begin{remrk}
The above form of the GEXIT function naturally conforms with the generalized
area theorem. Thus, we are able to write the GEXIT curve and use the
MAP bounding technique.\end{remrk}
\begin{lemma}
Assume that all the channel families are the same%
\footnote{Note that for the case of different channel families, one can still
compute the $i$th GEXIT function as a function of the common parameter
$\epsilon$.%
}, i.e., $\mathtt{h}_{i}=\mathtt{h}$. The $i$th GEXIT function is
given by\begin{align*}
\mathsf{G}_{i}(\mathtt{h}) & =\sum_{z}p(z)\int_{\underline{v}}\mathsf{a}_{i,z}(\underline{v})\kappa_{i,z}(\underline{v})\text{d}\underline{v}\end{align*}
where $\mathsf{a}_{i,z}$ is the distribution of the vector $\Phi_{i}$
given $Z_{i}=z$, $\underline{v}$ is a vector of length $|\mathcal{Z}|$
in the $(|\mathcal{Z}|-1)$-dimensional probability simplex and the
GEXIT kernel (for $i$ and $z$) is%
\footnote{$p(y_{i}|z)$ is dependent on $\mathtt{h}_{i}$ and hence is dependent
on $\epsilon$.%
}\[
\kappa_{i,z}(\underline{v})=\frac{\int_{-\infty}^{\infty}\frac{\partial}{\partial\epsilon}p(y_{i}|z)\log_{2}\left\{ \frac{\sum_{z'}v_{[z']}p(y_{i}|z')}{v_{[z]}p(y_{i}|z)}\right\} \text{d}y_{i}}{\int_{-\infty}^{\infty}\sum_{z}p(z)\frac{\partial}{\partial\epsilon}p(y_{i}|z)\log_{2}\left\{ \frac{\sum_{z'}p(z')p(y_{i}|z')}{p(z)p(y_{i}|z)}\right\} \text{d}y_{i}}.\]
\end{lemma}
\begin{proof}
Suppose the initial state is $S_{0}$, we start by writing\begin{align}
H(X_{1}^{n}|Y_{1}^{n},S_{0}) & =H(Z_{1}^{n}|Y_{1}^{n},S_{0})\nonumber \\
 & =H(Z_{i}|Y_{1}^{n},S_{0})+H(Z_{\sim i}|Y_{1}^{n},Z_{i},S_{0}).\label{eq:1}\end{align}

For simplicity of notation, we drop $S_{0}$ in all the expressions
although the dependency on $S_{0}$ is always implied. From (\ref{eq:gih1hn})
and (\ref{eq:1}), it is clear that

\begin{align}
\mathsf{G}_{i}(\mathtt{h})= & \frac{\partial}{\partial\mathtt{h}_{i}}H(Z_{i}|Y_{1}^{n}).\label{eq:Gih}\end{align}
We also have\begin{align}
 & H(Z_{i}|Y_{1}^{n})=H(Z_{i}|Y_{i},\Phi_{i}(Y_{\sim i}))\nonumber \\
 & \phantom{H(Z_{i}|Y_{1}^{n})}=-\int_{\phi_{i}}\int_{y_{i}}\sum_{z_{i}}p(z_{i})p(\phi_{i}|z_{i})p(y_{i}|z_{i})\log_{2}\!\left\{ \!\frac{p(z_{i}|\phi_{i})p(y_{i}|z_{i})}{\sum_{z_{i}^{'}}p(z_{i}^{'}|\phi_{i})p(y_{i}|z_{i}^{'})}\!\right\} \text{d}y_{i}\text{d}\phi_{i}\label{eq:2}\end{align}
where (\ref{eq:2}) follows from the Bayes' theorem and the fact that
\begin{equation}
p(z_{i},\phi_{i},y_{i})=p(z_{i},\phi_{i})p(y_{i}|\phi_{i},z_{i})=p(z_{i})p(\phi|z_{i})p(y_{i}|z_{i}).\label{eq:3}\end{equation}
Note that (\ref{eq:3}) is true since $Y_{i}$ and $\Phi_{i}(Y_{\sim i})$
are independent given $Z_{i}$, i.e., $Y_{i}\rightarrow Z_{i}\rightarrow\Phi_{i}(Y_{\sim i}).$ 

Taking derivative and using $p(z_{i}|\phi_{i})=p(z_{i}|y_{\sim i})$,
we get%
\footnote{One can verify that the terms obtained by taking derivative with respect
to the channel inside the $\log_{2}$ vanish.%
}\begin{align*}
\mathsf{G}_{i}(\mathtt{h}) & =\sum_{z_{i}}p(z_{i})\int_{\phi_{i}}p(\phi_{i}|z_{i})\int_{y_{i}}\frac{\text{d}}{\text{d}\mathtt{h}_{i}}p(y_{i}|z_{i})\log_{2}\left\{ \sum_{z_{i}'}\frac{p(z'_{i}|y_{\sim i})p(y_{i}|z'_{i})}{p(z_{i}|y_{\sim i})p(y_{i}|z_{i})}\right\} \text{d}y_{i}\text{d}\phi_{i}\\
 & =\sum_{z}p(z)\int_{\underline{v}}\mathsf{a}_{i,z}(\underline{v})\kappa_{i,z}(\underline{v})\text{d}\underline{v}.\end{align*}
where\begin{align*}
\kappa_{i,z}(\underline{v}) & =\int_{y_{i}}\frac{\text{d}}{\text{d}\mathtt{h}_{i}}p(y_{i}|z)\log_{2}\left\{ \frac{\sum_{z'}v_{[z']}p(y_{i}|z')}{v_{[z]}p(y_{i}|z)}\right\} \text{d}y_{i}\\
 & =\int_{y_{i}}\frac{\partial}{\partial\epsilon}p(y_{i}|z)\log_{2}\left\{ \frac{\sum_{z'}v_{[z']}p(y_{i}|z')}{v_{[z]}p(y_{i}|z)}\right\} \text{d}y_{i}/\frac{\partial\mathtt{h}_{i}}{\mbox{\ensuremath{\partial}}\epsilon}.\end{align*}

Finally, by seeing that

\begin{align*}
\frac{\partial\mathtt{h}_{i}}{\partial\epsilon} & =\frac{\partial H(Z_{i}|Y_{i}(\epsilon))}{\partial\epsilon}\\
 & =\sum_{z}\int_{y_{i}}p(z)\frac{\partial}{\partial\epsilon}p(y_{i}|z)\log_{2}\left\{ \frac{\sum_{z'}p(z')p(y_{i}|z')}{p(z)p(y_{i}|z)}\right\} \text{d}y_{i}.\end{align*}
we obtain the result.\end{proof}
\begin{remrk}
For erasure noise and the GEC in particular, $\mathtt{h}=H(Z|Y)=\epsilon H(Z)$
(scaling $\epsilon$ by $H(Z))$ and since in this case \[
\kappa_{i,z}(\underline{v})=\frac{1}{H(Z)}\log_{2}\left\{ 1+\frac{\sum_{z'\neq z}v_{[z']}}{v_{[z]}}\right\} ,\]
$\mathsf{G}(\mathtt{h})=\frac{h(\epsilon)}{H(Z)}$ (scaling $h(\epsilon)$
by $\frac{1}{H(Z)}$) where $h(\epsilon)$ is the EXIT function for
the GEC.
\end{remrk}

\begin{remrk}
\label{rem:TwoExtremes}At $\sigma=0$ for AWGN case (or at $\epsilon=0$
for erasure noise), $\mathtt{h}=0$ and $\mathsf{a}_{i,z}$ is {}``delta
at $\underline{v}=e_{[z]}"$ where $e_{[z]}$ is the standard basis
vector. At this extreme, $\mathsf{G}(0)=0$ since $\kappa_{i,z}(\underline{v})=0$.
At the other extreme $\sigma\rightarrow\infty$ (or at $\epsilon=1$
for erasure noise), $\mathtt{h}=H(Z)$ (e.g., $1.5$ for the dicode
channel) and $\mathsf{G}(\mathtt{h})=1$ since in this case $\mathsf{a}_{i,z}$
is {}``delta at $v_{[z']}=p(z')\,\,\forall z'$''. 
\end{remrk}

\subsubsection{BP-GEXIT curve (with AWGN)}

In this section, we are particularly interested in computing the BP-GEXIT
function for ISI channels with AWGN. In this case, let $\Phi_{i}^{\text{BP},\ell}$
denote the extrinsic estimate of $Z_{i}$ at the $\ell$th round of
joint BP decoding. If $\Phi_{i}^{\text{BP},\ell}$ is used instead
of $\Phi_{i}$ in the above formulas then one has the BP-GEXIT (at
the $\ell$th round) $\mathsf{G}^{\text{BP},\ell}$ in a similar manner
to \cite{Measson-it09} and the overall BP-GEXIT $\mathsf{G}^{\text{\text{BP}}}(\mathtt{h})=\lim_{\ell\rightarrow\infty}\mathsf{G}^{\text{BP},\ell}(\mathtt{h})$.
Also, notice that the two extremes in Remark \ref{rem:TwoExtremes}
still apply when the BP decoder is used instead of the MAP decoder. 

Next, AWGN implies that $p(y_{i}|z)=\frac{1}{\sqrt{2\pi\sigma^{2}}}e^{-\frac{(y_{i}-z)^{2}}{2\sigma^{2}}}$
and then $\frac{\partial}{\partial\epsilon}p(y_{i}|z)=((y_{i}-z)^{2}-\sigma^{2})p(y_{i}|z)$.
Therefore, the corresponding $i$th BP-GEXIT is $\mathsf{G}_{i}^{\text{BP,\ensuremath{\ell}}}(\mathtt{h})=\frac{A}{B}$
where\begin{align*}
A & =\sum_{z}p(z)\int_{\underline{v}}\mathsf{a}_{i,z}^{\text{BP},\ell}(\underline{v})\int_{-\infty}^{\infty}p(y_{i}|z)\left\{ \frac{(y_{i}-z)^{2}}{\sigma^{2}}-1\right\} \log_{2}\left\{ \sum_{z'}\frac{v_{[z']}}{v_{[z]}}e^{\frac{(z'-z)(2y_{i}-z-z')}{2\sigma^{2}}}\right\} \text{d}y_{i}\text{d}\underline{v}\end{align*}
and\begin{align*}
B & =\sum_{z}p(z)\int_{-\infty}^{\infty}p(y_{i}|z)\left\{ \frac{(y_{i}-z)^{2}}{\sigma^{2}}-1\right\} \log_{2}\left\{ \sum_{z'}\frac{p(z')}{p(z)}e^{\frac{(z'-z)(2y_{i}-z-z')}{2\sigma^{2}}}\right\} \text{d}y_{i}.\end{align*}

In the limit of $\ell\rightarrow\infty$, one can run the DE for ISI
channels \cite{Kavcic-it03} to obtain the DE-FP and compute the quantities
$A$ and $B$ at this FP. With some abuse of notation, let $\mathsf{a}^{(\ell)},\mathsf{b}^{(\ell)},\mathsf{c}^{(\ell)}$
and $\mathsf{d}^{(\ell)}$ denote the average density of the bit-to-check,
check-to-bit, bit-to-trellis and trellis-to-bit messages, respectively
(see Fig. \ref{fig:JointGraph}), at iteration $\ell$ with initial
values (at $\ell=0$) being $\Delta_{0}$, the delta function at $0$.
Also, let $\mathsf{n}$ denote the density of channel noise. The DE
update equation for joint BP decoding of a general binary-input ISI
channels is\begin{align*}
\mathsf{a}^{(\ell)} & =\mathsf{d}^{(\ell-1)}\oast\lambda(\mathtt{b}^{(\ell-1)}),\\
\mathsf{b}^{(\ell)} & =\rho(\mathsf{a}^{(\ell)}),\\
\mathsf{c}^{(\ell)} & =L(\mathsf{b}^{(\ell)}),\\
\mathsf{d}^{(\ell)} & =\Gamma(\mathsf{c}^{(\ell)},\mathsf{n})\end{align*}
where for a density $\mathsf{x}$, $\lambda(\mathsf{x})=\sum_{i}\lambda_{i}\mathsf{x}^{\oast(i-1)},\rho(\mathsf{x})=\sum_{i}\rho_{i}\mathsf{x}^{\boxast(i-1)}$
and $L(\mathsf{x})=\sum_{i}L_{i}\mathsf{x}^{\oast i}$. The operators
$\oast$ and $\boxast$ are the standard density transformations used
in \cite[p. 181]{RU-2008}. The map $\Gamma(\cdot,\cdot)$ is not
easy to compute in closed form for general trellises and often one
needs to resort to the Monte Carlo methods (i.e., running the windowed
BCJR algorithm with window parameter $W$ on a long enough trellis
- see details in \cite{Kavcic-it03}) to give the estimates. A similar
method was used to upper bound the MAP threshold for turbo codes over
BMS channels \cite{Measson-isit05}.

The denominator $B$ can be computed either by numerical integration
or by Monte Carlo methods. Meanwhile, the numerator $A$ involves
in the quantity $v_{[z]}=p\left(Z_{i}=z|\mathtt{T}_{i}^{\ell}\right)$
where $\mathtt{T}_{i}^{\ell}$ denotes the computation tree of depth
$\ell$, rooted at index $i$, which includes all channel and code
constraints associated with $\ell$ iterations of decoding. This computation
tree $\mathtt{T}_{i}^{\ell}$ excludes the tree root $y_{i}$ and
is implied by the decoding schedule in the DE equation. The quantity
$v_{[z]}$, due to complications from the trellis, is not easy to
obtain in closed form. However, one can readily compute $v_{[z]}$
as an extra output of the BCJR algorithm (already used in DE) as\[
v_{[z]}\propto\sum_{s_{i},s_{i-1}:Z_{i}=z}\alpha_{i-1}(s_{i-1})\cdot\gamma_{i}(s_{i-1},s_{i})\cdot\beta_{i}(s_{i}).\]
where $\gamma_{i}(s_{i-1},s_{i})$ is probability of the input $x_{i}$
that corresponds to the transition from state $s_{i-1}$ (at time
index $i-1)$ to state $s_{i}$ at (time index $i$) given the computation
tree $\mathtt{T}_{i}^{\ell}$. Here, $\alpha_{i}(\cdot)$ and $\beta_{i}(\cdot)$
are the standard forward and backward state probabilities in the BCJR
algorithm. Note that the scaling constant can be chosen so that $\sum_{z}v_{[z]}=1$.

\subsection{Upper Bound for the MAP Threshold}

As briefly discussed before, the above-mentioned GEXIT curve naturally
follows the area theorem\[
\int_{\mathtt{h}^{\text{MAP}}}^{H(Z)}\mathsf{G}(\mathtt{h})\text{d}\mathtt{h}=\int_{0}^{H(Z)}\mathsf{G}(\mathtt{h})\text{d}\mathtt{h}=\mathtt{r}.\]

One can also apply \cite[Lm. 4]{Measson-it09} to the BMS channel
from $Z_{1}^{n}$ to $Y_{1}^{n}$ and obtains\[
\frac{\partial H(Z_{i}|Y_{1}^{n})}{\partial\mathtt{h}_{i}}\leq\frac{\partial H(Z_{i}|Y_{i},\Phi_{i}^{\text{BP},\ell})}{\partial\mathtt{h}_{i}}.\]
Consequently, by invoking (\ref{eq:Gih}), one has the optimality
of the MAP decoder in the sense that $\mathsf{G}(\mathtt{h})\leq\mathsf{G}^{\text{BP}}(\mathtt{h})$.
Therefore, one can use the discussed bounding technique, i.e., by
finding the largest value $\bar{\mathtt{h}}^{\text{MAP}}$ such that
the area under the BP-GEXIT curve equals the code rate,

\[
\int_{\bar{\mathtt{h}}^{\text{MAP}}}^{H(Z)}\mathsf{G}^{\text{BP}}(\mathtt{h})\text{d}\mathtt{h}=\mathtt{r},\]
to obtain the MAP upper bound $\bar{\mathtt{h}}^{\text{MAP}}\geq\mathtt{h}^{\text{MAP}}$
(as $\int_{\bar{\mathtt{h}}^{\text{MAP}}}^{H(Z)}\mathsf{G}^{\text{BP}}(\mathtt{h})\text{d}\mathtt{h}=\int_{\mathtt{h}^{\text{MAP}}}^{H(Z)}\mathsf{G}(\mathtt{h})\text{d}\mathtt{h}\leq\int_{\mathtt{h}^{\text{MAP}}}^{H(Z)}\mathsf{G}^{\text{BP}}(\mathtt{h})\text{d}\mathtt{h}$). 

For example, the BP-GEXIT curve for the $(3,6)$-regular LDPC code
over an AWGN dicode channel with $a(D)=(1-D)/\sqrt{2}$ following
the analysis in Section \ref{sub:GEXIT-Curves} is shown in Fig. \ref{fig:BP-GEXIT-curve}.
In this case, $\mathtt{h}^{\text{BP}}(3,6)\approx0.851\pm0.001$ (the
corresponding%
\footnote{We adopt the convention that $\sigma$ is the SNR threshold measured
in dB.%
} $\sigma^{\text{BP}}(3,6)\approx1.703\pm0.001$ dB) while $\bar{\mathtt{h}}^{\text{MAP}}(3,6)\approx0.920\pm0.001$
(or $\bar{\sigma}^{\text{MAP}}(3,6)\approx0.959\pm0.001$ dB). Similarly,
for the $(5,10)$-regular LDPC code, one has $\mathtt{h}^{\text{BP}}(5,10)\approx0.716\pm0.001$
and $\bar{\mathtt{h}}^{\text{MAP}}(5,10)\approx0.931\pm0.001$. The
corresponding thresholds measured in dB can be found in Table. \ref{tab:thresholds}.%
\begin{table}
\begin{centering}
\begin{tabular}{|c||c|c|c|c|c|c|}
\hline 
$(l,r)$- & \multicolumn{3}{c|}{DEC} & \multicolumn{3}{c|}{Dicode AWGN}\tabularnewline
\cline{2-7} 
regular & $\epsilon^{\text{BP}}$ & $\bar{\epsilon}^{\text{MAP}}$ & $\epsilon^{\text{SIR}}$ & $\sigma^{\text{BP}}$ & $\bar{\sigma}^{\text{MAP}}$ & $\sigma^{\text{SIR}}$\tabularnewline
\hline
\hline 
$(3,6)$ & \textsf{$0.5689$} & \textsf{$0.6387$} & $0.6404$ & \textsf{$1.073$} & $0.959$ & $0.823$\tabularnewline
\hline 
$(5,10)$ & \textsf{$0.4647$} & \textsf{$0.6404$} & \textbf{$0.6404$} & \textsf{$3.032$} & $0.834$ & $0.823$\tabularnewline
\hline
\end{tabular}
\par\end{centering}

\caption{\label{tab:thresholds}Threshold estimates of $(l,r)$-regular ensembles
over the DEC and dicode AWGN channel. For AWGN noise, the thresholds
are measured in dB.}
\vspace{-7mm}
\end{table}

\begin{figure}
\begin{centering}
\begin{minipage}[c][1\totalheight][t]{0.485\textwidth}%
\begin{center}
\includegraphics[scale=0.65]{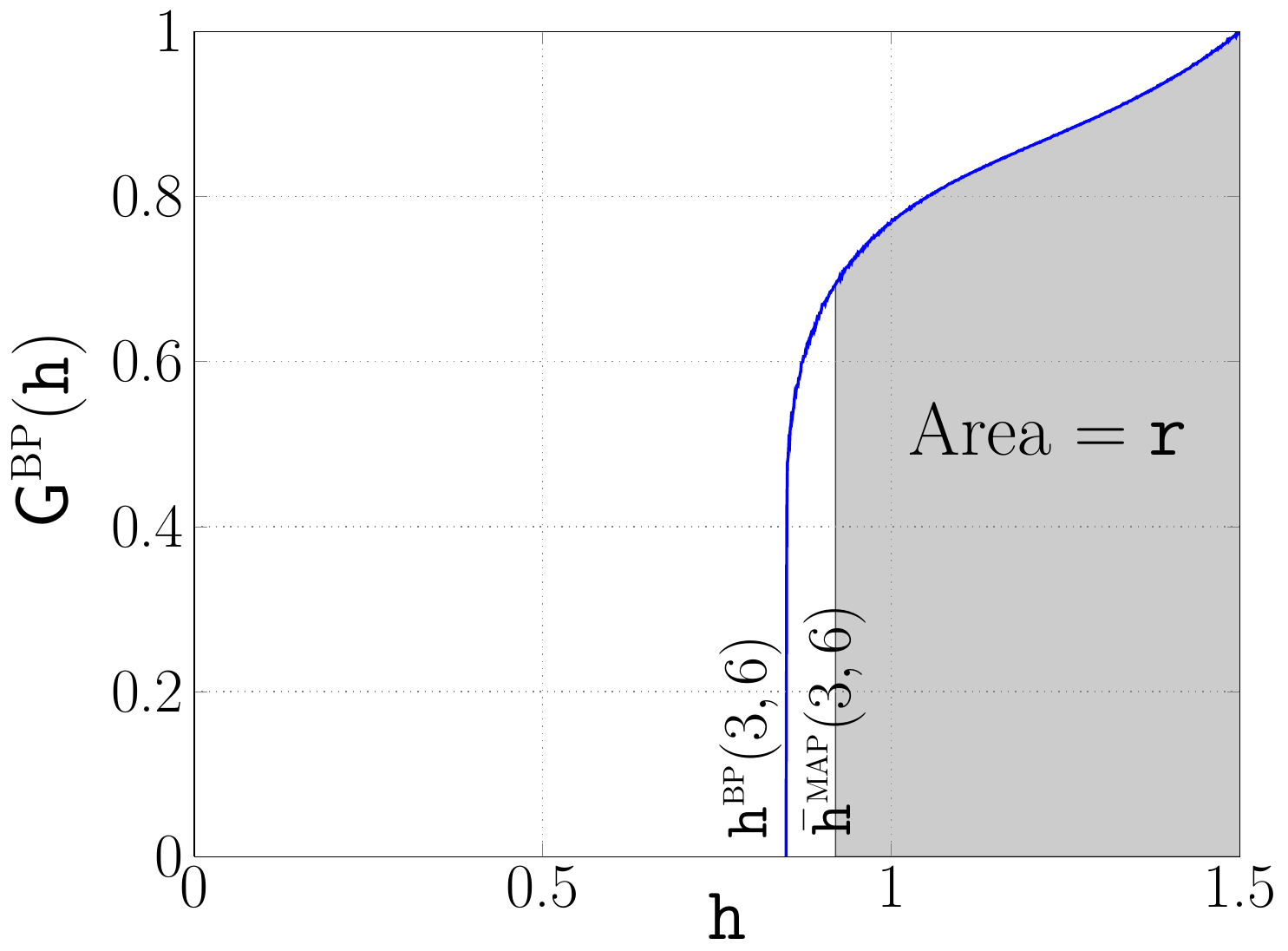}
\par\end{center}%
\end{minipage}\hspace{0.5cm}%
\begin{minipage}[c][1\totalheight][t]{0.485\textwidth}%
\begin{center}
\includegraphics[scale=0.65]{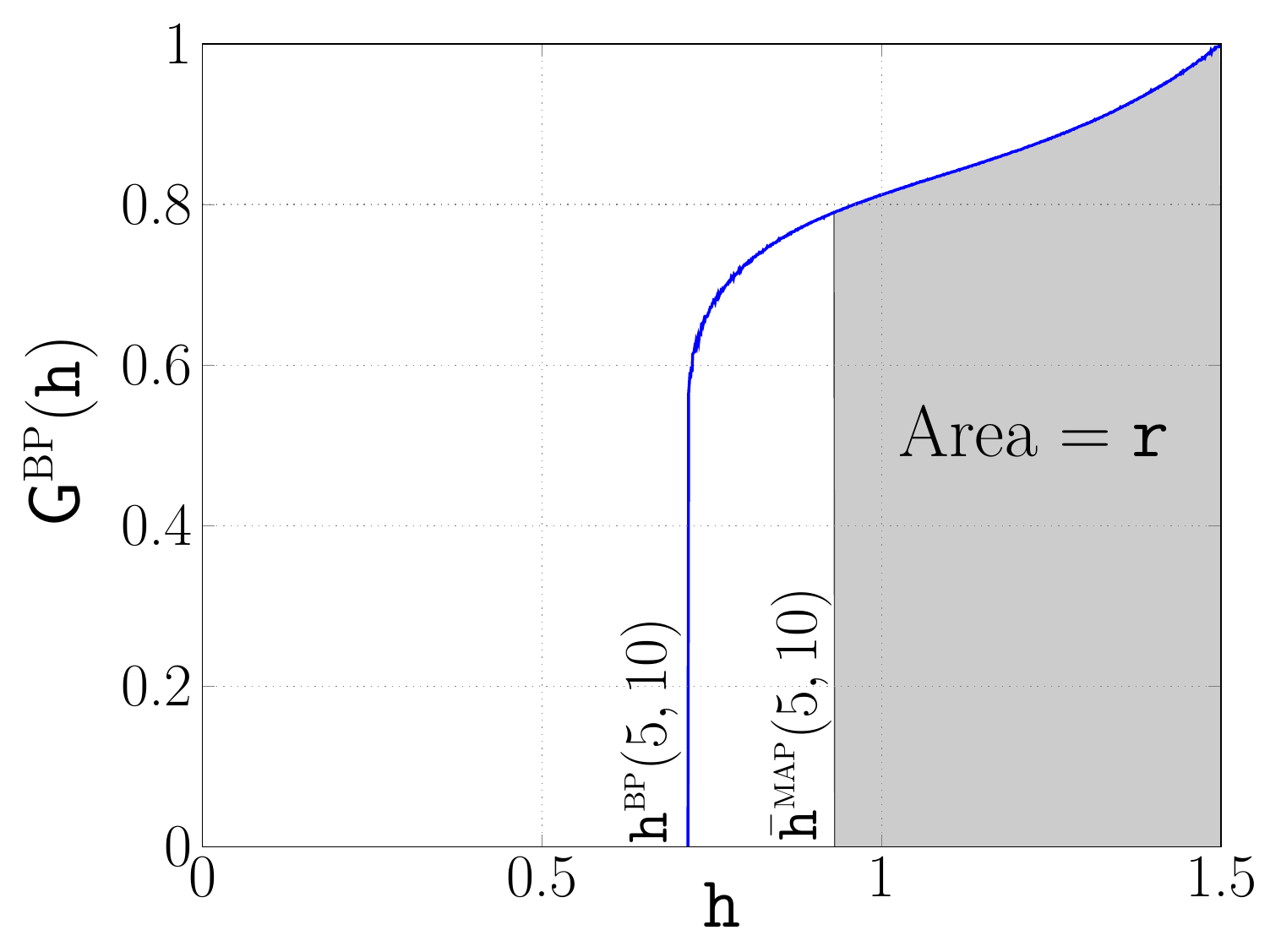}
\par\end{center}%
\end{minipage}
\par\end{centering}

\caption{\label{fig:BP-GEXIT-curve}The BP-GEXIT curve for $(3,6)$-regular
and $(5,10)$-regular LDPC codes over an AWGN dicode channel with
$a(D)=(1-D)/\sqrt{2}$. The upper bound $\bar{\mathtt{h}}^{\text{MAP}}$
is obtained by setting the area under the BP-GEXIT curve (the shaded
region) equal to the code rate.}

\end{figure}

\subsection{\label{sub:DESCISI}Spatially-Coupled Codes on the ISI Channels}

Consider the $(l,r,L)$ spatially-coupled ensemble. For the ISI channels,
the DE equation for this ensemble can be obtained from the protograph
chain in a similar manner to the case of memoryless channels discussed
in \cite{Lentmaier-it10}. For each $i,j\in[1-\hat{l},L+\hat{l}]$,
let $\mathsf{a}_{i\rightarrow j}^{(\ell)}$ (and $\mathsf{b}_{i\leftarrow j}^{(\ell)}$)
denote the average density of the messages from bit nodes at position
$i$ to check nodes at position $j$ (and the other way around)%
\footnote{For $i\notin[1,L]$, set $\mathsf{a}_{i\rightarrow j}^{(\ell)}=\Delta_{+\infty},$
the delta function at $+\infty$. %
}. With all the initial message densities (at $\ell=0$) being $\Delta_{0}$,
the DE update equation (for all $i\in[1,L]$) is

\begin{align*}
\mathsf{a}_{i\rightarrow j}^{(\ell)} & =\mathsf{d}_{i}^{(\ell-1)}\oast\left\{ \oastlimits_{j'\in[i-\hat{l},i+\hat{l}]\setminus j}\mathsf{b}_{i\leftarrow j'}^{(\ell-1)}\right\} ,\forall j\in[i-\hat{l},i+\hat{l}],\\
\mathsf{b}_{i\leftarrow j}^{(\ell)} & =\boxastlimits_{i'\in[j-\hat{l},j+\hat{l}]\setminus i}\mathsf{a}_{i'\rightarrow j}^{(\ell)},\forall j\in[i-\hat{l},i+\hat{l}],\\
\mathsf{c}_{i}^{(\ell)} & =\oastlimits_{j'\in[i-\hat{l},i+\hat{l}]}\mathsf{b}_{i\leftarrow j'}^{(\ell)},\\
\mathsf{d}_{i}^{(\ell)} & =\Gamma(\mathsf{c}_{i}^{(\ell)},\mathsf{n})\end{align*}
where $\oastlimits_{j\in\{j_{1},\ldots,j_{t}\}}\mathsf{x}_{j}$ and
$\boxastlimits_{i\in\{i_{1},\ldots,i_{t}\}}\mathsf{x}_{i}$ denote
the operations $\mathsf{x}_{j_{1}}\oast\mathsf{x}_{j_{2}}\oast\ldots\oast\mathsf{x}_{j_{t}}$
and $\mathsf{x}_{i_{1}}\boxast\mathsf{x}_{i_{2}}\boxast\ldots\boxast\mathsf{x}_{i_{t}}$,
respectively.

\subsection{Simulation Results}

In this section, we start with the $(l,r,L)$ circular ensemble obtained
by considering all the positions $i>L$ of the protograph chain to
be the same as position $i-L$ (similar to \cite{Kudekar-istc10}).
The order of bit transmissions is {}``left to right'' in each length-$L$
row and then start with the next row (in a total of $M$ rows, see
Fig. \ref{fig:JGraph(3,6,L)}). The $I\triangleq\max(\nu,l-1)$ first
bits in each row are known. This known bits will {}``break'' the
circular ensemble into the $(l,r,L-I)$ ensemble and also serve as
the pilot bits to fix the trellis state. As a consequence of this
fixing, one only needs to run the BCJR independently in each row and
this can be done in a parallel manner \cite{Narayanan-aller04,Soriaga-it07}. 

In our experiments, we conduct simulations over the AWGN dicode channel
with $a(D)=(1-D)/\sqrt{2}$ and memory $\nu=1$. First, we use the
DE in Sec. \ref{sub:DESCISI} to compute the BP thresholds of the
spatially-coupled coding scheme. The results in Fig. \ref{fig:BER-and-BP}
reveals that $\sigma^{\text{BP}}(3,6,22)$ is roughly $0.959\pm0.001$
dB and approximately the same as $\sigma^{\text{BP}}(3,6,44)$ whose
rate loss is smaller. Notice that this is also roughly $\bar{\sigma}^{\text{MAP}}(3,6)$
- the MAP threshold estimate of the underlying $(3,6)$-regular ensemble,
obtained by the bounding technique, and is a significant improvement
over $\sigma^{\text{BP}}(3,6)\approx1.703\pm0.001$ dB. This suggests
that threshold saturation occurs for regular ensembles. Since MAP
decoding of regular ensembles can achieve the SIR \cite{Bae-jsac09},
it also implies that one can universally approach the SIR of general
ISI channels using coupled codes with joint iterative decoding. To
support this, one can also see that for the $(5,10,44)$ ensemble
of the same rate as the $(3,6,22)$ one, the threshold $\sigma^{\text{BP}}(5,10,44)\approx0.834\pm0.001$
dB (which is also roughly $\bar{\sigma}^{\text{MAP}}(5,10))$ gets
very close to the signal-to-noise ratio (SNR) corresponding to the
SIR ($\sigma^{\text{SIR}}\approx0.823\pm0.001$ dB using the numerical
method in \cite{Arnold-icc01,Pfister-globe01}). 

Also shown in Fig. \ref{fig:BER-and-BP} is the bit error rate (BER)
versus SNR plot for the ensembles derived from the $(l,r,L)$ circular
ensembles of finite $M=502$ and $M=5000$. For each simulation, we
use $\mathtt{l}_{\text{outer}}=20$ channel updates and between two
such channel updates, we run $\mathtt{l}_{\text{inner}}=5$ BP iterations
on the code part alone. The curves labeled {}``target'' is the BER
for the bits at position $I+1$ (right after the known bits) in the
coupled chain while the curve labeled {}``overall'' is the overall
BER for all the positions $[I+1,L]$ together. One might expect that
the {}``overall'' BER will get closer to the {}``target'' BER
for large enough $M$ and large enough number of iterations. From
Fig. \ref{fig:BER-and-BP}, one can also observe that the {}``overall''
BER for $(3,6,22)$ and $M=5000$ keeps getting {}``closer'' to
the {}``target'' BER as SNR slightly increases. Those BER curves
are way to the left of $\epsilon^{\text{BP}}(3,6)$ - the BP threshold
for the underlying $(3,6)$-regular ensemble. %
\begin{figure}
\begin{centering}
\includegraphics[scale=0.85]{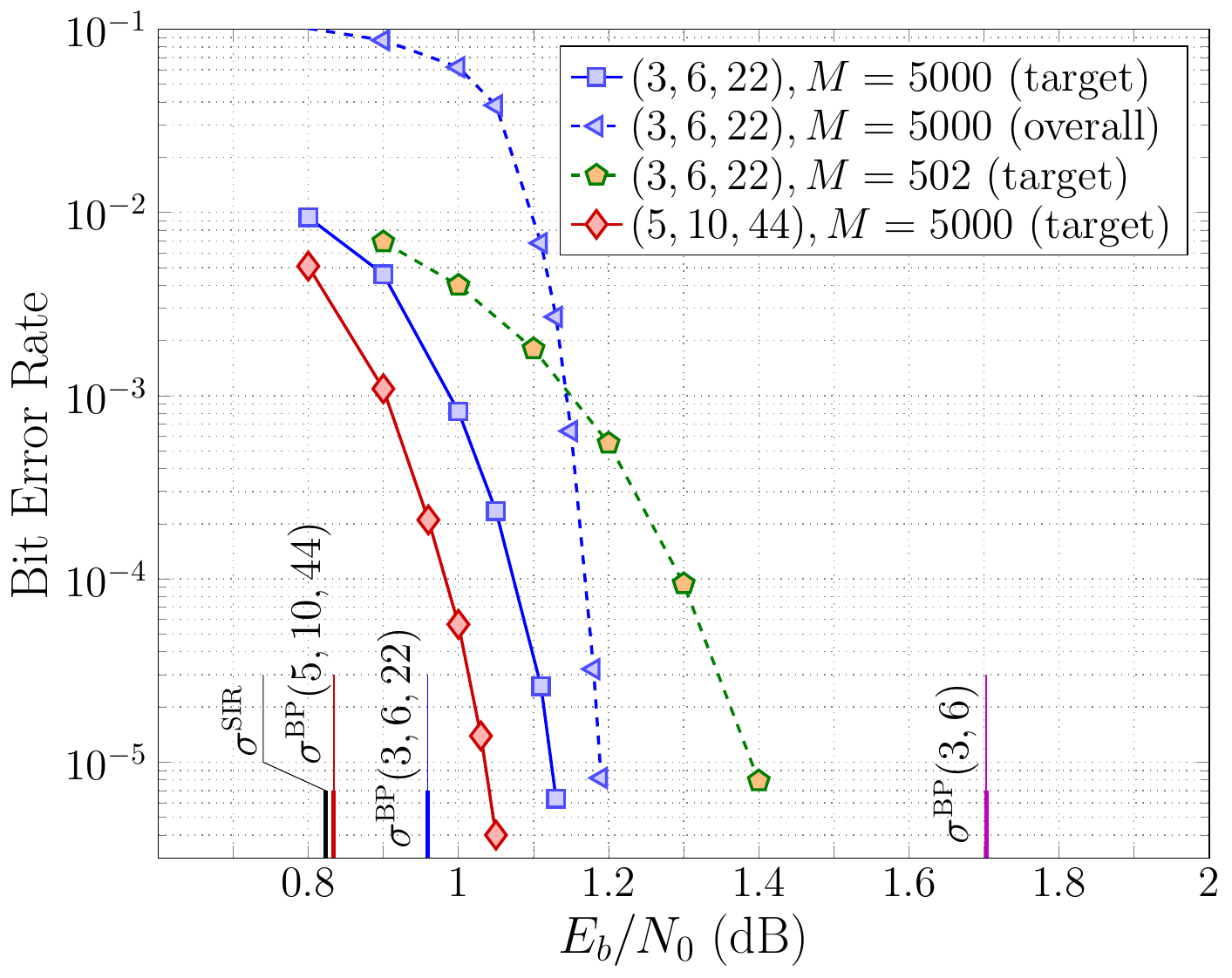}
\par\end{centering}

\caption{\label{fig:BER-and-BP}BER and BP thresholds for the $(3,6)$-regular
LDPC code, $(3,6,22)$ and $(5,10,44)$ spatially-coupled codes over
the AWGN dicode channel.}

\end{figure}

\section{Concluding Remarks}

In this paper, we consider binary communication over the ISI channels
and numerically show that the threshold saturation effect occurs on
both the DEC and dicode channel with AWGN. To do this, we construct
the EXIT and GEXIT curves that satisfy the area theorem and obtain
an upper bound on the threshold of the MAP decoder. This upper bound
is conjectured to be tight and, for the DEC, we show a numerical evidence
which strongly supports this conjecture. The observed threshold saturation
effect is valuable because by changing the underlying regular LDPC
ensemble, i.e., increasing the degrees according to a fixed code rate,
combined with the results of \cite{Bae-jsac09}, it is shown that
the joint BP decoding of spatially-coupled codes can universally approach
the SIR of the ISI channels.

Also, it has been known that the spatially-coupled codes (or LDPC
convolutional codes) inherit some other advantages such as the typical
minimum distance and the size of the smallest non-empty trapping sets
both growing linearly with the protograph expansion $M$ \cite{Mitchell-isit11}.
In addition, the convolutional structure of the codes allows one to
consider a windowed decoder like the one discussed in \cite{Iyengar-itsub10,Iyengar-isit11}.
All of these properties suggest that spatially-coupled codes may be
competitive in practice for systems with ISI.

\bibliographystyle{IEEEtran}

\end{document}